\newif\iflong\longfalse
\documentclass{lmcs}
\pdfoutput=1

\usepackage{lastpage}
\lmcsdoi{17}{3}{25}
\lmcsheading{}{\pageref{LastPage}}{}{}%
{Jan.~22,~2020}{Sep.~17,~2021}{}

\title[Modular coinduction up-to for H.O. languages via
  F.O. transition systems]
{Modular coinduction up-to for higher-order languages via
  first-order transition systems}

\author[J.-M.~Madiot]{Jean-Marie Madiot}
\address{INRIA, France}

\author[D.~Pous]{Damien Pous}
\address{Plume team, LIP, CNRS, ENS Lyon, Universit{\'e} de Lyon, France}

\author[D.~Sangiorgi]{Davide Sangiorgi}
\address{Universit{\`a} di Bologna, Italy; INRIA, France}

\usepackage[utf8x]{inputenc}
\usepackage{amssymb,amsmath,mathpartir,stmaryrd}
\usepackage{verbatim,color,stackrel}
\usepackage{enumitem}
\usepackage{tikz}
\usetikzlibrary{matrix,arrows}
\usepackage{hyperref}

\DeclareUnicodeCharacter{8656}{\ensuremath{\Leftarrow}}
\DeclareUnicodeCharacter{10214}{\ensuremath{\llbracket}}
\DeclareUnicodeCharacter{10215}{\ensuremath{\rrbracket}}
\DeclareUnicodeCharacter{8659}{\ensuremath{\Downarrow}}
\DeclareUnicodeCharacter{8617}{\ensuremath{\hookleftarrow}}
\DeclareUnicodeCharacter{8618}{\ensuremath{\hookrightarrow}}
\DeclareUnicodeCharacter{8701}{\ensuremath{\leftarrowtriangle}}
\DeclareUnicodeCharacter{8702}{\ensuremath{\rightarrowtriangle}}
\DeclareUnicodeCharacter{8669}{\ensuremath{\leadsto}}
\DeclareUnicodeCharacter{12296}{⟨}
\DeclareUnicodeCharacter{12297}{⟩}
\DeclareUnicodeCharacter{9001}{⟨}
\DeclareUnicodeCharacter{9002}{⟩}
\DeclareUnicodeCharacter{8871}{\ensuremath{\models}}
\DeclareUnicodeCharacter{8614}{\ensuremath{\mapsto}}
\DeclareUnicodeCharacter{183}{\ensuremath{\cdot}}
\DeclareUnicodeCharacter{10987}{\ensuremath{\Bot}}
\DeclareUnicodeCharacter{8871}{\ensuremath{\models}}
\DeclareUnicodeCharacter{8902}{\ensuremath{\star}}
\DeclareUnicodeCharacter{9733}{\ensuremath{\star}}
\DeclareUnicodeCharacter{10229}{\ensuremath{\longleftarrow}}
\DeclareUnicodeCharacter{10230}{\ensuremath{\longrightarrow}}

\DeclareUnicodeCharacter{120016}{\ensuremath{\mathcal{A}}}
\DeclareUnicodeCharacter{120017}{\ensuremath{\mathcal{B}}}
\DeclareUnicodeCharacter{120018}{\ensuremath{\mathcal{C}}}
\DeclareUnicodeCharacter{120019}{\ensuremath{\mathcal{D}}}
\DeclareUnicodeCharacter{120020}{\ensuremath{\mathcal{E}}}
\DeclareUnicodeCharacter{120021}{\ensuremath{\mathcal{F}}}
\DeclareUnicodeCharacter{120022}{\ensuremath{\mathcal{G}}}
\DeclareUnicodeCharacter{120023}{\ensuremath{\mathcal{H}}}
\DeclareUnicodeCharacter{120024}{\ensuremath{\mathcal{I}}}
\DeclareUnicodeCharacter{120025}{\ensuremath{\mathcal{J}}}
\DeclareUnicodeCharacter{120026}{\ensuremath{\mathcal{K}}}
\DeclareUnicodeCharacter{120027}{\ensuremath{\mathcal{L}}}
\DeclareUnicodeCharacter{120028}{\ensuremath{\mathcal{M}}}
\DeclareUnicodeCharacter{120029}{\ensuremath{\mathcal{N}}}
\DeclareUnicodeCharacter{120030}{\ensuremath{\mathcal{O}}}
\DeclareUnicodeCharacter{120031}{\ensuremath{\mathcal{P}}}
\DeclareUnicodeCharacter{120032}{\ensuremath{\mathcal{Q}}}
\DeclareUnicodeCharacter{120033}{\ensuremath{\mathcal{R}}}
\DeclareUnicodeCharacter{120034}{\ensuremath{\mathcal{S}}}
\DeclareUnicodeCharacter{120035}{\ensuremath{\mathcal{T}}}
\DeclareUnicodeCharacter{120036}{\ensuremath{\mathcal{U}}}
\DeclareUnicodeCharacter{120037}{\ensuremath{\mathcal{V}}}
\DeclareUnicodeCharacter{120038}{\ensuremath{\mathcal{W}}}
\DeclareUnicodeCharacter{120039}{\ensuremath{\mathcal{X}}}
\DeclareUnicodeCharacter{120040}{\ensuremath{\mathcal{Y}}}
\DeclareUnicodeCharacter{120041}{\ensuremath{\mathcal{Z}}}

\theoremstyle{definition}
\newtheorem{lemma}[thm]{Lemma}
\newtheorem{example}[thm]{Example}
\newtheorem{theorem}[thm]{Theorem}
\newtheorem{corollary}[thm]{Corollary}

\newtheorem{definition}[thm]{Definition}
\theoremstyle{remark}
\newtheorem{remark}[thm]{Remark}

\newcommand{\simE}{\mathrel{\sim^{\rm{e}}}}
\newcommand{\wbE}{\mathrel{\approx^{\rm{e}}}}
\newcommand{\wb}{\approx}

\newcommand{\st}{\, \mbox{s.t.} \,}

\newcommand{\Dwa}{\Downarrow}           
\newcommand{\dwa}{\downarrow}           

\def\res#1{{\boldsymbol \nu} #1\:}   

\newcommand{\holei}[1]{[\cdot]_{#1}}  

\newcommand{\LaI}{\Lambda{R}}
\newcommand{\LaIu}{\LaI^1}

\newcommand{\LaN}{\Lambda{N}}
\newcommand{\LaNu}{\LaN^1}

\newcommand{\simEBv}[1]{\mathrel{{\approx_{#1}}}}

\newcommand{\starred}[1]{\mathrel{#1^\star}}

\newcommand{\simEB}{\approx}

\def\X{{\mathcal X}}

\def\Y{{\mathcal Y}}

\newcommand{\Xv}[1]{\mathrel{{\X_{#1}}}}

\newcommand{\contexthole}{ [ \cdot  ] }      
\newcommand{\holE}{\contexthole}  

\newcommand{\rn}[1]{%
  \ifmmode 
    \mathchoice
      {\mbox{\sc #1}}
      {\mbox{\sc #1}}
      {\mbox{\small\sc #1}}
      {\mbox{\tiny\uppercase{#1}}}%
  \else
    {\sc #1}%
  \fi}

\newcommand{\Lao}{\Lambda^0}  

\def\sub#1#2{\{\raisebox{.5ex}{\small$#1$}\! / \!\mbox{\small$#2$}\}}

\def\defi{\eqdef}

\newcommand{\chain}{\raisebox{.9ex}[0ex][0ex]{\mbox{\footnotesize$\frown$}}}
\def\RR{\mathrel{\mathcal R}} 
\def\S{{\mathcal S}}          
\def\SS{\mathrel{\mathcal S}} 

\def\EE{{\mathcal E}} 

\newcommand{\arrr}[1]{\mathrel{\stackrel{{\;\;#1\;\;}}{\mbox{\rightarrowfill}}}}

\newcommand{\Act}{\mbox{\it Act}} 
\newcommand{\pr}{\mbox{\it Pr}} 

\newcommand{\arrpi}[1]{\:\stackrel{{#1}}{\longrightarrow}_\pi\:}

\renewcommand{\simEBv}[1]{\mathrel{{\approx^\textsf{env}_{#1}}}}

\renewcommand{\simEB}{\approx^\textsf{env}}

\newcommand{\arr}{\longrightarrow} 
\newcommand{\warr}{\Longrightarrow}
\newcommand{\rarr}{\longleftarrow} 

\newcommand{\arrT}{\longmapsto} 
\newcommand{\warrT}{\Longmapsto}

\newcommand{\ts}[1]{\mathrel{\stackrel{#1}{\arr}}}           
\newcommand{\ws}[1]{\mathrel{\stackrel{#1}{\warr}}}            
\newcommand{\rts}[1]{\mathrel{\stackrel{#1}{\rarr}}}           
\newcommand{\tpi}[1]{\mathrel{\stackrel{#1}{\arrT}}_{\pi}}      
\newcommand{\wpi}[1]{\mathrel{\stackrel{#1}{\warrT}}_{\pi}}     

\renewcommand{\arrpi}[1]{\:\stackrel{{#1}}{\arrT}_\pi\:}

\newcommand{\eqdef}{\triangleq}

\newcommand{\refeq}{\equiv} 

\newcommand{\names}[1]{\ensuremath{\mathrm{n}(#1)}}
\newcommand{\fnames}[1]{\ensuremath{\mathrm{fn}(#1)}}
\newcommand{\bnames}[1]{\ensuremath{\mathrm{bn}(#1)}}

\newcommand{\new}{ν}
\newcommand{\out}[1]{\ensuremath{\overline{#1}}}
\newcommand{\subst}[2]{\ensuremath{\{#1/#2\}}}

\newcommand{\cbn}{\arrT_{\mathsf{n}}}

\newcommand{\icbv}{\arrT_{\mathsf{R}}}
\newcommand{\wcbn}{\warrT_{\mathsf{n}}}

\newcommand{\wicbv}{\warrT_{\mathsf{R}}}
\newcommand{\vcontexts}{\mathbb C_v}

\newcommand{\dom}[1]{\mathsf{dom}(#1)}
\newcommand{\newloc}{νℓ\,}

\newcommand{\cod}{\ensuremath{{\tt cod}}}
\newcommand*{\sep}{;}
\newcommand*{\get}[1]{\mathsf{get}_{#1}}%
\newcommand*{\set}[1]{\mathsf{set}_{#1}}%
\newcommand*{\getset}[1]{\mathsf{getset}(#1)}%
\newcommand*{\lget}[1]{\mathsf{get}_{#1}}%
\newcommand*{\lset}[1]{\mathsf{set}_{#1}}%
\newcommand{\inc}{\textsf{incr}_ℓ}
\newcommand{\tes}{\textsf{test}_ℓ}
\newcommand{\tesp}{\textsf{test}_ℓ}
\newcommand{\tl}{\tilde{ℓ}}%
\newcommand{\incz}{\textsf{incr}_0}%
\newcommand{\tesz}{\textsf{test}_0}%

\newcommand{\R}{\mathrel{\mathcal R}}
\renewcommand{\S}{\mathrel{\mathcal S}}

\newcommand*{\relprogress}[1]{\stackbin[#1]{}{\leadsto}}
\newcommand*{\relprogressfootnote}[1]{\stackbin[{}^{#1}]{}{\leadsto}}

\newcommand{\funprogress}[1]{\stackrel{#1}{\leadsto}} 
\newcommand*{\cofix}[1]{\nu(#1)}

\newcommand{\bb}{{\bf{p}}}
\newcommand{\strongb}{\ensuremath{\mathbf{sp}}}
\newcommand{\weakb}{\ensuremath{\mathbf{wp}}}
\newcommand{\expb}{\ensuremath{\mathbf{ep}}}

\newcommand{\uptotrans}{\textsf{star}}
\newcommand{\uptoexp}{\ensuremath{f_{≳}}}

\newcommand{\isubst}{\ensuremath{\mathsf{isub}}} 
\newcommand{\bsubst}{\ensuremath{\mathsf{bsub}}} 
\newcommand{\str}{\ensuremath{\mathsf{str}}}
\newcommand{\icout}{\mathcal{C}_o}
\newcommand{\icnew}{\mathcal{C}_{ν}}
\newcommand{\icbang}{\mathcal{C}_{!}}
\newcommand{\icpar}{\mathcal{C}_{∣}}
\newcommand{\icplus}{\mathcal{C}_{+}}
\newcommand{\icgplus}{\mathcal{C}_{g+}}
\newcommand{\icin}{\mathcal{C}_i}

\newcommand{\env}{\mathsf{str}}
\newcommand{\term}{{𝓒}}
\newcommand{\eval}{{𝓒_{\textsf{e}}}}
\newcommand{\alloc}{\mathsf{store}}
\newcommand{\weak}{\mathsf{w}}
\newcommand{\id}{\mathsf{id}}

\newcommand{\arrR}[1]{\mathrel{\stackrel{{\;\;#1\;\;}}{\mbox{\rightarrowfill}}}}

\newcommand{\HIDEPROOF}[1]{}
\definecolor{blu}{rgb}{.4,.4,1}

\begin{document}

\maketitle

\begin{abstract}
  The bisimulation proof method can be enhanced by employing
  \emph{`bisimulations up-to'} techniques.  A comprehensive theory of
  such enhancements has been developed for first-order (i.e.,
  CCS-like) labelled transition systems (LTSs) and bisimilarity, based
  on abstract fixed-point theory and compatible functions.
  
  We transport this theory onto languages whose bisimilarity and LTS
  go beyond those of first-order models.  The approach consists in
  exhibiting fully abstract translations of the more sophisticated
  LTSs and bisimilarities onto the first-order ones.  This allows us
  to reuse directly the large corpus of up-to techniques that are
  available on first-order LTSs.  The only ingredient that has to be
  manually supplied is the compatibility of basic up-to techniques
  that are specific to the new languages.  We investigate the method
  on the $\pi$-calculus, the $\lambda$-calculus, and a (call-by-value)
  $\lambda$-calculus with references.
\end{abstract}

\section{Introduction}
\label{s:intro}

One of the keys for the success of bisimulation is its associated
proof method, whereby to prove two terms equivalent, one exhibits a
relation containing the pair and one proves it to be a bisimulation.
The bisimulation proof method can be enhanced by employing relations
called \emph{`bisimulations up-to'} \cite{San98MFCS,Len98,SanPous,RotBR13}; see 
\cite{PousS19} for a historical perspective.
These need not be bisimulations; they are simply \emph{contained in} a
bisimulation.  Such techniques have been widely used in languages for
mobility such as $\pi$-calculus or higher-order languages such as the
$\lambda$-calculus, or Ambients (e.g.,
\cite{Lassen98relationalreasoning,MZ05,SW01}).

Several forms of bisimulation enhancements have been introduced:
`bisimulation up to bisimilarity' \cite{Mil89} where the derivatives
obtained when playing bisimulation games can be rewritten using
bisimilarity itself; `bisimulation up to transitivity' where the
derivatives may be rewritten using the up-to relation \cite{San98MFCS}; `bisimulation
up to context' \cite{Sangiorgi94}, where a common context may be removed
from matching derivatives.  Further enhancements may exploit the
peculiarities of the definition of bisimilarity on certain classes of
languages: e.g., the up-to-injective-substitution techniques of the
$\pi$-calculus~\cite{JeffreyR99,SW01}, techniques for shrinking or
enlarging the environment in languages with information hiding
mechanisms (e.g., existential types, encryption and decryption
constructs~\cite{AbadiG98,SumiiP07:seal,SumiiP07:type}), frame
equivalence in the psi-calculi~\cite{Parrow}, or higher-order
languages~\cite{KoutavasW06,Las98}. Lastly, it is important to notice
that one often wishes to use \emph{combinations} of up-to
techniques. For instance, up-to-context alone does not appear to be
very useful; its strength comes out in association with other
techniques, such as up-to-bisimilarity or up-to-transitivity.

The main problem with up-to techniques is proving their soundness
(i.e.\ ensuring that any `bisimulation up-to' is contained in
bisimilarity).  In particular, the proofs of complex combinations of
techniques can be difficult or, at best, long and tedious. Moreover, if one
modifies the language or the up-to technique, the entire proof has to
be redone from scratch.  Indeed the soundness of some up-to techniques
is quite fragile, and may break when such variations are made.  For
instance, up-to-bisimilarity usually fails for weak bisimilarity, and
in certain languages the combination of up-to-bisimilarity and
up-to-context fails while the two techniques are sound when taken
separately.

This problem has been the motivation for the development of a theory
of enhancements, summarised in \cite{SanPous}. Expressed in the
general fixed-point theory on complete lattices, this theory has been
fully developed for both strong and weak bisimilarity, in the case of
first-order labelled transition systems (LTSs) where transitions
represent pure synchronisations among processes.  In this framework,
up-to techniques are represented using {\em compatible} functions,
whose class enjoys nice algebraic properties (an earlier variant, with similar properties,
is that of
\emph{respectful}  functions \cite{San98MFCS}).  This allows one to
derive complex up-to techniques algebraically, by composing simpler
techniques by means of a few operators.

Only a small part of the theory has been transported onto other forms
of transition systems, on a case by case basis.  Transferring the whole theory would be a
substantial and non-trivial effort.  Moreover it might have limited
applicability, as this work would probably have to be based on
specific shapes for transitions and bisimilarity (a wide range of
variations exist, e.g., in higher-order languages).

Here we explore a different approach to the transport of the theory of
bisimulation enhancements onto richer languages.  The approach
consists in exhibiting fully abstract translations of the more
sophisticated LTSs and bisimilarities onto first-order LTSs and
bisimilarity. This allows us to import directly the existing theory
for first-order bisimulation enhancements onto the new languages. Most
importantly, the schema allows us to combine up-to techniques for the
richer languages.  The only additional ingredient that has to be
provided manually is the soundness of some up-to techniques that are
specific to the new languages. This typically includes the
up-to-context techniques, since those contexts are not first-order.

Our hope is that the method proposed here will make it possible to
obtain a single formalised library about up-to techniques, that can be
reused for a wide range of calculi: currently, all existing
formalisations of such techniques in a proof assistant are specific to
a given calculus: $\pi$-calculus~\cite{Miller,Hirschkoff97}, the
psi-calculi~\cite{Parrow}, or a miniML language~\cite{HurNDV13}.

We consider three languages in this paper: the $\pi$-calculus, the
call-by-name $\lambda$-calculus, and an imperative call-by-value
$\lambda$-calculus (a call-by-value $\lambda$-calculus with
references).
We focus on weak bisimilarity, whose theory is more involved than that
of strong bisimilarity: some the congruence properties break or
require subtle proofs, and these differences between the strong and
weak case are magnified when it comes to enhancements of the
bisimulation proof method.

When we translate a transition system into a first-order
one, the grammar for the transition labels can be complex (e.g.\
include terms, labels, or contexts). What nevertheless makes these
systems `first-order' is that these labels are taken as syntactic
atomic objects, that may only be checked for syntactic equality. In
other words, the bisimulation games which we play on those first-order
LTS are just the plain and standard ones, without any side conditions
on free names, alpha-conversion, or semantical comparison of
higher-order values.
Leifer and Milner~\cite{LeiferMilner00}, using 
contexts as  first-order
labels, derive, from any
appropriate reactive system, a first-order LTS for which strong
bisimilarity is a congruence; this approach however does not handle
weak bisimilarity or up-to techniques.

Full abstraction of a translation does not imply
that all desirable or expected up-to techniques come for free: sometimes the translations 
 have to be
designed with care to be useful. We shall
see this with the $\pi$-calculus, where \emph{early bisimilarity} can
be handled properly, but where the natural and fully abstract
adaptation of the translation for \emph{late bisimilarity} does not
provide us with satisfactory up-to techniques (see
Remark~\ref{rem:early:late}).
In the same manner, our translation for the $\lambda$-calculus has similarities with
the first-order translation of environmental bisimulation
in~\cite{KoutavasLS11}, but the latter, while fully
abstract, would disallow important up-to techniques (see
Remark~\ref{r:cancellation}).

Forms of up-to-context have already been derived for the languages we
consider in this
paper~\cite{Lassen98relationalreasoning,envbisim,SW01}. The
corresponding soundness proofs are difficult (especially in
$\lambda$-calculi), and require a mix of induction (on contexts) and
coinduction (to define bisimulations).  Recasting up-to-context within
the theory of bisimulation enhancements has several advantages. First,
this allows us to combine this technique with other techniques,
directly. Second, congruence (or substitutivity) of bisimilarity
becomes a corollary of the compatibility of the up-to-context function
(in higher-order languages these two kinds of proofs are usually hard
and very similar).  And third,
this
allows us to decompose the up-to-context function into
smaller pieces, essentially one for each operator of the language,
yielding more modular proofs, also allowing, if needed, to rule out
those contexts that do not preserve bisimilarity (e.g., input prefix
in the $\pi$-calculus).

The translation of the $\pi$-calculus LTS into a first-order LTS
follows the schema of abstract machines for the $\pi$-calculus (e.g.,
\cite{Tur96}) in which the issue of the choice of fresh names is
resolved by ordering the names and indexing the processes with a name
that represents an upper bound to the names occurring in the process.
Various forms of bisimulation enhancements have appeared in papers on
the $\pi$-calculus or dialects of it. A translation of higher-order
$\pi$-calculi into first-order processes has been proposed by Koutavas
et al.~\cite{KoutavasHennessy12}. While the shape of our translations
of $\lambda$-calculi is similar, our LTSs differ since they are
designed to recover the theory of bisimulation enhancements. In
particular, using the fully abstract LTSs
from~\cite{KoutavasHennessy12}, does not make it possible to use up-to
context techniques (essentially by the same argument as the one in
Remark~\ref{r:cancellation}).
In the $\lambda$-calculus, limited forms
of up-to techniques have been developed for applicative bisimilarity,
where the soundness of up-to-context is still an open
problem~\cite{Lassen98relationalreasoning,Las98}. More powerful
versions of up-to-context exist for forms of bisimilarity on open
terms; e.g., open bisimilarity or head-normal-form
bisimilarity~\cite{Lassen99}.  Currently, the form of bisimilarity
for closed higher-order terms that allows the richest range of up-to
techniques is environmental
bisimilarity~\cite{KoutavasLS11,envbisim}. However, even in this
setting, the proofs of combinations of up-to techniques are usually
long and non-trivial.
Our translation of higher-order terms to first-order terms is designed
to recover environmental bisimilarity and to simplify the tasks of
proving up-to techniques and of combining them.

Open bisimilarity or normal-form bisimilarities have also been used to
avoid quantification over function
arguments~\cite{San94,Lassen98relationalreasoning,KoutavasW06}, and
powerful and compositional approaches to up-to techniques have been
developed in such settings~\cite{BiernackiLP19,BiernackiLP20}. Yet the
resulting LTSs still manipulate binders and require fresh
instantiations of them, much in the same way as the $\pi$-calculus does; as
such, as discusssed above, the theory of first-order LTSs is not
directly applicable.

This paper is an extended version of the conference
paper~\cite{MADIOTPS14}, with all proofs. That paper introduced a
notion `compatibility up-to' in order to assemble several up-to
techniques whose soundness depends on each other, in a `mutually
coinductive' way. This led to the development of the notion of
``companion''~\cite{PousCompanion}, reviewed in Section~\ref{s:back},
as a cleaner way of doing the same thing and more. We rewrote the
current paper to make use of the companion in place of `compatibility
up-to'. This also allowed us to established novel results for
additional proof refactoring (Lemmas~\ref{l:env:eval:values}
and~\ref{l:val:nonval}).

In Section~\ref{s:back} we review the theory of first-order
bisimulation and up-to techniques. In Sections~\ref{s:pi} to
\ref{s:icbv} we treat the $\pi$-calculus, the (pure) call-by-name
$\lambda$-calculus, and the imperative call-by-value
$\lambda$-calculus, respectively.  In Section~\ref{ss:exa}, we show an example of
how the wide spectrum of up-to techniques made available via our
translations allows us to simplify relations needed in bisimilarity
proofs, facilitating their description and reducing their size.

\subsection*{Notations}
We let $\R$, $\S$ range over binary relations and we often write
$x\R y$ for $(x,y)\in{\R}$. Given two relations $\R,\S$, we write
$\R\S$ for their relational composition, i.e.,
${\R\S}\eqdef\{(x,z)\mid \exists y, x\R y \land y\S z\}$, $\R^+$ for
the transitive closure of $\R$, and $\R^*$ for its reflexive
transitive closure. We use the standard arrow notation $\mapsto$ to
denote functions when the domain is clear, and $⇒$ for logical
implication; other arrow notations will be introduced as we go along.

In languages defined from a grammar, a \emph{context} $C$ of arity
$n \in \mathbb N$ is a term with numbered holes $\holei 1$, …,
$\holei n$, where each hole $\holei i$ can appear any number of times in
$C$. We write $C[P_1,\dots,P_n]$ for the application of such a context
to $n$ terms $P_1,\dots,P_n$ of the language.

\section{First-order bisimulation and up-to techniques}
\label{s:back}

As explained in the introduction, the results in this section are not
new:  we review general-purpose tools that we exploit to prove
soundness of up-to techniques. These tools were obtained in several
steps:  respectfulness is from~\cite{San98MFCS}; we refined it into the notion of
\emph{compatibility up-to} in the conference version of this
paper~\cite{MADIOTPS14}, and this refinement eventually led to the
notion of the companion~\cite{PousCompanion} and to the associated tools
we exploit here.

A \emph{first-order Labelled Transition System}, briefly LTS, is a
triple {$(\pr, \Act, \longrightarrow ) $} where $\pr$ is a non-empty
set of states (or processes), $\Act$ is the set of {\em actions} (or
\emph{labels}), and $\mathord{\longrightarrow} \subseteq \pr \times\Act\times
\pr $ is the \emph{transition relation}.  We use $P,Q, R$ to range
over the processes of the LTS, and $\mu$ to range over the labels in
$\Act$, and, as usual, write ${P} \arrr\mu Q$ when ${(P, \mu, Q)} \in
{\longrightarrow} $.  We assume that $\Act$ includes a special action
$\tau$ that represents an internal activity of the processes.  We
derive bisimulation from the notion of \emph{progression} between
relations.

\begin{definition}\label{d:strongprogression}
  We define the monotone function $\strongb$ on relations on processes
  of an LTS:  
  \begin{align*}
    \strongb(𝓡) ≜ 
    \{ (P,Q) \mid~~
    &\left(\forall P' ~~\forall μ ~{~}~ P \ts{μ} P' ~~⇒~~ \exists Q' ~~ Q\ts{μ}Q' \wedge P' \R Q'\right)\\
    \wedge
    &\left(\forall Q' ~~\forall μ ~{~}~ Q \ts{μ} Q' ~~⇒~~ \exists P' ~~ P\ts{μ}P' \wedge P' \R Q'\right)
      \}
  \end{align*}
  We say that $𝓡$ \emph{strongly progresses to} $𝓢$, written
  $𝓡\relprogress{\strongb} 𝓢$, if ${\R} \subseteq \strongb(\S)$. A
  relation $𝓡$ is a \emph{strong bisimulation} if
  $𝓡\relprogress{\strongb}𝓡$; and \emph{strong bisimilarity},
  $\sim$, is the union of all strong bisimulations.
\end{definition}



To define weak progression we need weak transitions, defined as usual:
first, $\smash{P \ts{\hat{μ}}P'}$ means $P \arrr\mu P'$ or ($\mu =\tau$
and $P=P'$); and $\ws{\hat{μ}} $ is $\Longrightarrow \ts{\hat{μ}}
\Longrightarrow$ where $\Longrightarrow$ is the reflexive transitive
closure of $\ts{τ}$.

\begin{definition}\label{d:weakprogression}
  We define the monotone function $\weakb$ on relations on processes
  of an LTS:
  \begin{align*}
    \weakb(𝓡) ≜ 
    \{ (P,Q) \mid~~
      &\left(\forall P' ~~\forall μ ~{~}~ P \ts{μ} P' ~~⇒~~ \exists Q' ~~ Q\ws{\hat{μ}}Q' \wedge P' \R Q'\right)\\
      \wedge
      &\left(\forall Q' ~~\forall μ ~{~}~ Q \ts{μ} Q' ~~⇒~~ \exists P' ~~ P\ws{\hat{μ}}P' \wedge P' \R Q'\right)
    \}
  \end{align*}
  We say that $𝓡$ \emph{weakly progresses to} $𝓢$, written
  $𝓡\relprogress{\weakb} 𝓢$, if ${\R} \subseteq \weakb(\S)$. A
  relation $𝓡$ is a \emph{weak bisimulation} if
  $𝓡\relprogress{\weakb}𝓡$; and \emph{weak bisimilarity}, $\wb$, is
  the union of all weak bisimulations.
\end{definition}


Below we summarise the ingredients of the theory of bisimulation
enhancements for first-order LTSs from~\cite{SanPous} that will be
needed in the sequel. We use $f$ and $g$ to range over monotone
functions on relations over a fixed set of states. Each such function
represents a potential up-to technique; only the \emph{sound}
functions, however, qualify as up-to techniques:

\begin{definition}
  A function $f$ is \emph{sound for $\sim$} if $𝓡\relprogress{\strongb}
  f(𝓡)$ implies $𝓡⊆ {\sim}$, for all $𝓡$; similarly, $f$ is
  \emph{sound for $\approx$} if $𝓡\relprogress{\weakb} f(𝓡)$ implies $𝓡⊆
  {\approx}$, for all $𝓡$.
\end{definition}

Unfortunately, the class of sound functions does not enjoy good
algebraic properties. In particular, the composition and the pairwise
union of two sound functions are not necessarily sound~\cite[Section
6.3.3]{SanPous}.
As a remedy to this, the subset of \emph{compatible} functions has
been proposed.
The concepts in the remainder of the section can be instantiated
with both strong and weak bisimilarities; we thus use $\bb$ to range
over ${\strongb}$ or ${\weakb}$.

\begin{definition}
  We write $f\funprogress{\bb} g$ when $f\circ \bb \subseteq \bb \circ g$.
  A monotone function $f$ on relations is \emph{$\bb$-compatible} if
  $f \funprogress{\bb} f$.
\end{definition}

In other terms, $f\funprogress{\bb} g$ when $𝓡\relprogress{\bb} 𝓢$
implies $f(𝓡)\relprogress{\bb} g(𝓢)$ for all $𝓡$ and $𝓢$.


\begin{lemma}
  \label{l:compatible:sound}
  If $f$ is $\strongb$-compatible, then $f$ is sound for $\sim$;
  if $f$ is $\weakb$-compatible, then $f$ is sound for $\approx$.
\end{lemma}
\HIDEPROOF{
\begin{lemma}[Compatible is sound]
  If $f\funprogress{\bb} f$ and $𝓡\relprogress{\bb} f(𝓡)$ then ${\R}⊆\cofix{\bb}$.
\end{lemma}
\begin{proof}
  $f\funprogress{\bb} f⊆f^ω$ implies by induction $f^n \funprogress{\bb} f^ω$ which
  implies $f^ω \funprogress{\bb} f^ω$.  From that and $𝓡\relprogress{\bb} f(𝓡)$ we
  deduce $f^ω(𝓡)\relprogress{\bb} f^ω(f(𝓡))⊆f^ω(𝓡)$ so $f^ω(𝓡)⊆\cofix{s}$ and
  since $f$ is extensive, $𝓡⊆f^ω(𝓡)⊆\cofix{s}$.
\end{proof}
}
Simple examples of compatible functions are the identity function
$\id$ and the function mapping any relation onto bisimilarity (strong
or weak case, depending on the considered case). This means that
$(𝓡↦{\sim})$ is $\strongb$-compatible, and $(𝓡↦{\approx})$ is
$\weakb$-compatible. In addition, $(𝓡↦{\sim})$ is also a
useful $\weakb$-compatible function.
The class of
compatible functions is closed under function composition and union
(where the union $∪F$ of a set of functions $F$ is the point-wise
union mapping $\R$ to $\bigcup_{f∈F}f(𝓡)$), and thus under
$\omega$-iteration (where the $\omega$-iteration $f^ω$ of a function
$f$ maps $\R$ to $\bigcup_{n∈\mathbb N}f^n(𝓡)$).
For example  $(𝓡↦({𝓡}\cup{\sim}))$ is
$\strongb$- and $\weakb$-compatible.

Other examples of compatible functions are typically contextual
closure functions, or \emph{up-to-context}, mapping a relation into its
closure w.r.t.\ a given set of contexts. Not all context
closures are compatible: their compatibility must be established
separately for each LTS. For such functions, the following lemma shows
that the compatibility of up-to-context implies the congruence of
(strong or weak) bisimilarity.

\begin{lemma}
  \label{l:compatible:closure}
  If $f$ is $\strongb$-compatible, then $f(\sim)⊆ {\sim}$; similarly
  if $f$ is $\weakb$-compatible, then $f(\approx)⊆ {\approx}$.
\end{lemma}
\HIDEPROOF{
\begin{proof}
  Since $\cofix{\bb}\relprogress{\bb}\cofix{\bb}$,
  $f(\cofix{\bb})\relprogress{\bb} f(\cofix{\bb})$ hence $f(\cofix{\bb})$ is a
  simulation.
\end{proof}
}

\HIDEPROOF{
\begin{lemma}\label{l:soupe}
  If for all $j$, $f_j \funprogress{\bb} (\cup_if_i)^ω$ then $(\cup_if_i)^ω
  \funprogress{\bb} (\cup_if_i)^ω$.
\end{lemma}
\begin{proof}
  We write $F=\cup_if_i$. Since $f_j⊆F^ω$ we know $F \funprogress{} F^ω$ as
  $\funprogress{}$ preserves union.  We conclude as $\funprogress{}$ preserves
  iteration and $(F^ω)^ω = F^ω$.
\end{proof}
}

Certain closure properties for compatible functions however only hold
in the strong case. The main example is the \emph{chaining operator}
$\chain$, which implements pointwise relational composition:
\[ f \chain g \: (𝓡) \defi f(𝓡)\; g(𝓡) \] where the juxtaposition
$𝓡\;{\S}$ of two relations $𝓡$ and $\S$ and
denotes their relational composition.  Using chaining we can obtain
the compatibility of the `up-to-transitivity' function, mapping a
relation $\R$ onto its reflexive and transitive closure $\starred{𝓡}$.
Another important example of compatible function in the strong case is
`up-to-strong-bisimilarity' ($𝓡↦{\sim\R\sim}$), which is
also compatible in the weak case.

In contrast, the counterpart of this latter function in the weak case,
$𝓡↦{\approx\R\approx}$, is unsound. This is a major drawback in up-to
techniques for weak bisimilarity, which can be partially overcome by
resorting to the \emph{expansion} relation $≳$ \cite{A-KHe92,SaMi92}
(a refinement of expansion is the contraction relation
\cite{Sangiorgi17}). Expansion is an asymmetric refinement of weak
bisimilarity whereby $P ≳ Q $ holds if $P $ and $Q$ are bisimilar and,
in addition, $Q$ is at least as efficient as $P$, in the sense that
$Q$ is capable of producing the same activity as $P$ without ever
performing more internal activities (the $\tau$-actions). More
precisely, the associated progression function is $\expb$ defined
below, and $≳$ is the union of all $𝓡$ such that $𝓡 ⊆ \expb(𝓡)$.
\begin{align*}
  \expb(𝓡) ≜ 
  \{ (P,Q) \mid~~
  &\left(\forall P' ~~\forall μ ~{~}~ P \ts{μ} P' ~~⇒~~ \exists Q' ~~ Q\ts{\hat{μ}}Q' \wedge P' \R Q'\right)\\
  \wedge
  &\left(\forall Q' ~~\forall μ ~{~}~ Q \ts{μ} Q' ~~⇒~~ \exists P' ~~ P\ws{}\ts{\mu}\ws{}P' \wedge P' \R Q'\right)
    \}
\end{align*}
Up-to-expansion yields a function ($𝓡 ↦ {≳\R≲}$) that is contained in
a $\weakb$-compatible function, and which can be freely combined with
any $\weakb$-compatible function, yielding for instance the
`up-to-expansion-and-contexts' technique. More sophisticated up-to
techniques can be obtained by carefully adjusting the interplay
between visible and internal transitions, and by taking into account
termination hypotheses~\cite{SanPous}.

Some further compatible functions are the functions $\strongb$ and
$\weakb$ themselves (indeed a function $f$ is $\bb$-compatible if
$f \circ \bb \subseteq \bb \circ f $, hence trivially $f$ can be
replaced with $\bb$ itself). Intuitively, the use of $\strongb$ and
$\weakb$ as up-to techniques means that, in a diagram-chasing
argument, the two derivatives need not be related; it is sufficient
that the derivatives of such derivatives be related. Accordingly, we
sometimes call functions $\strongb$ and $\weakb$ \emph{unfolding}
functions. We will use $\strongb$ in the example in
Section~\ref{ss:exa} and $\weakb$ in Sections~\ref{s:cbn}
and~\ref{s:icbv}, when proving the $\weakb$-compatibility of the
up-to-context techniques. Note that up-to-context functions are the
only ones that need to be proved compatible separately for each LTS;
in this section all other functions mentioned as compatible are so for
every first-order LTS.

\subsection{Companion}

We say that $f$ is \emph{below} $g$, and write $f \subseteq g$, when
$f(\RR) \subseteq g(\RR)$ for every relation $\RR$. Any function below
a sound function is sound as well. Similarly, if $f\subseteq g$ and
$g(\sim)\subseteq {\sim}$ then $f(\sim)\subseteq {\sim}$.

In general a function below a compatible function need not be itself
compatible. However, it turns out that there is a largest compatible
function, which is called the \emph{companion}
of~$\bb$~\cite{PousCompanion}, defined as the pointwise union of all
$\bb$-compatible functions:
\[t_{\bb} ~~\eqdef \bigcup_{f\;\funprogress{\bb}\;f} f \]%
In the following, we  generally omit the subscript or superscript
$\bb$ when  clear from the context. Since $t$ is itself compatible we
can deduce from Lemmas~\ref{l:compatible:sound} and
\ref{l:compatible:closure} that
if
$f\subseteq t_\strongb$, then $f$ is sound for $\sim$ and
$f(\sim)\subseteq {\sim}$. 
Similarly in the weak case: if
$f\subseteq t_\weakb$ then $f$ is sound for $\approx$ and
$f(\approx) \subseteq {\approx}$. 

The identity function $\id$ and the function $\bb$ itself are below
$t$. The fact that function composition preserves compatibility is
reflected by the idempotence of $t$, i.e.\ $t \circ t = t$.  Since the
companion is idempotent and contains all compatible functions, every
bisimulation proof up to a certain combination of compatible functions
can be presented as a bisimulation up to the companion.  Although
this observation does not make such proofs fundamentally easier, it
slightly simplifies their presentation: the precise combination of
up-to techniques does not have to be made explicit. This is typically
extremely convenient in a proof assistant.

\subsection{Tools for validating up-to techniques}

The companion makes it possible to perform bisimulation proofs up to
arbitrary combinations of functions that are known to be below it. In
concrete languages, we thus have to prove that the functions
associated to up-to techniques such as up-to-context, are indeed below
the companion. By definition of the companion, given a function $f$,
an obvious way to prove $f\subseteq t$ consists in proving that $f$ is
compatible, i.e.\footnote{%
  The notation $\leadsto$ stands for $\funprogress{\bb}$ (on
  functions) and $\relprogressfootnote{\bb}$ (on relations); this overloading
  is explained below.
}, $f\leadsto f$. This is however quite restrictive in
practice, because many useful functions are not compatible by
themselves, they are only contained in a compatible function, which is
often hard to express explicitly. (Very much like bisimulation up-to,
which can be small and convenient to work with while the concrete
bisimulations lying over them can be large or hard to express.)

Seeing the companion as a coinductive object, one can in fact relax the
requirement $f\leadsto f$ ``$f$ is compatible'' into $f\leadsto F(f)$
``$f$ is compatible up to $F$'', where $F$ is a \emph{second order}
technique~\cite{pous:aplas07:clut}. For instance,
\begin{enumerate}
\item if $f \leadsto (f \cup g)$ for some $g \subseteq t$, then
  $h\eqdef f \cup t$ is compatible so that $f\subseteq t$. This means
  we can freely exploit a function $g$ already known to be below the
  companion when establishing a progression about $f$.
\item if $f \leadsto f^2$, then $f^\omega$ is compatible and
  contains $f$. This means we can use $f$ twice in a row when
  establishing a progression about $f$.  By a similar argument,
  $f \leadsto f^\omega$ also entails $f\subseteq t$, meaning we can
  actually use $f$ as many times as required.
\item for all sets $F$ of functions such that for all $f\in F$,
  $f \leadsto (\bigcup F \cup t)^\omega$, then
  $(\bigcup F \cup t)^{(\omega^\omega)}$ is compatible
  (where $h^{(\omega^\omega)}$ is $h ∪ h^ω ∪ (h^ω)^ω ∪ ((h^ω)^ω)^ω ∪ …$),
  so that all
  functions in $F$ are below the companion.  This intuitively makes it
  possible to reason by `mutual coinduction' in order to prove that a
  family of up-techniques is valid. 
\end{enumerate}
Leaving the companion aside, the last item above was in fact named
``compatibility up-to'' in the previous version of this
work~\cite{MADIOTPS14}. This idea was simplified
in~\cite{PousCompanion}, by defining the second-order function
$B(g) \eqdef \bigcup_{f\leadsto g} f$.  Indeed, the notation
$f \funprogress{\bb} g$, which was an apparent overloading of the
progression operator $\leadsto$, can now be seen as the regular
progression operator associated to the function $B$.

This function $B$ also has a companion written $T$, which is a
monotone function satisfying the following properties, for all
monotone functions $f$:
\begin{itemize}
\item if $f \leadsto T(f)$ then $f\subseteq t$;
\item $f \subseteq T(f)$, $t \subseteq T(f)$, and $T(f) = T(f)^2 = T(f)^\omega$.
\end{itemize}
The first point is just the fact that every function compatible up to
$T$ lies below the companion (like every bisimulation up to $t$ is
contained in bisimilarity). The second point tells us that given a
family $F$ of functions, $T(\cup F)$ actually contains all potential
combinations of functions in $F$ and functions below $t$.

The three examples of compatible functions up-to listed above can thus
be seen as particular instances of compatible functions up to $T$.  In
particular, the last item, which we will use repeatedly to prove that
up-to-context techniques are valid in the first-order LTSs we present,
can be generalised as follows:
\begin{enumerate}
\item[(3')] for all sets $F$ of functions such that for all $f\in F$,
  $f \leadsto T(\bigcup F)$, every function in $F$ is below $t$.
\end{enumerate}

\begin{remark}[On respectfulness]\label{remark:on:respectfulness}
  In the first modular treatment of up-to techniques for
  bisimilarity~\cite{San98MFCS,SW01}, the notion of \emph{respectful}
  function was used: a monotone function $f$ is respectful if for all
  $\R,\S$ such that ${\R}\subseteq{\S}$ and ${\R}\leadsto{\S}$, we
  have $f(\R)\leadsto f(\S)$. Every compatible function is respectful,
  but the converse is not true. The hypothesis ${\R}\subseteq{\S}$ was
  actually added in the definition of respectfulness to ease proofs
  about up-to-context, which typically lead to respectful functions
  that are not compatible. However, this difference between compatible
  and respectful functions  disappears when considering the
  companion: the largest compatible function and the largest
  respectful function coincide~\cite{PousCompanion}, so that focusing
  on the simpler notion of compatibility does not prevent us from
  using certain up-to techniques, in the end.

  In practice, proofs of up-to techniques based on respectfulness can
  be adapted to the compatibility setting as follows. Suppose we try
  to prove $f\leadsto T(f)$ for a specific function $f$, i.e., to
  prove that ${\R}\leadsto{\S}$ entails $f(\R)\leadsto T(f)(\S)$.  The
  missing assumption ${\R}\subseteq{\S}$ is in general useful for
  those cases where we obtain process derivatives related via $\R$
  rather than~$\S$. Respectfulness makes it possible to conclude
  directly in those cases, since
  \begin{align*}
    {\R}\subseteq {\S} \subseteq t(\S) \subseteq T(f)(\S)\enspace.
  \end{align*}
  With compatibility, we can use the up-to-unfolding technique: we have
  \begin{align*}
    {\R}\subseteq \bb (\S) \subseteq t(\S) \subseteq T(f)(\S)\enspace,
  \end{align*}
  where the first inclusion is just the assumption ${\R}\leadsto{\S}$.
\end{remark}

\section{The π-calculus}
\label{s:pi}

We let letters $a,b$ range over a set of \emph{names}.
We recall the syntax for $\pi$-calculus processes ($P$) and transition
labels ($\mu$):
\begin{align*}
  P& ::= 0 \: ∣ \: a(b).P\:  ∣ \: \out a b.P\: ∣\:  P|P\:  ∣ \: \res b
     P\:  ∣ \: {!} P\\
  μ &::= τ ∣ a b ∣ \out a b ∣ \out a (b) 
\end{align*}
The name $b$ is bound is $P$ in constructs $a(b).P$ and $\res b P$.
The early operational
semantics is described by the rules for $\tpi{}$ in
Figure~\ref{f:pi}. We write $\fnames{Q}$ for the free names in $Q$,
defined as usual. The names $\names{\mu}$ of $\mu$ are defined as
$\names{\out a b}≜\names{a b}≜\names{\out a(b)}≜\{a,b\}$ and
$\names{\tau}≜\emptyset$ and the bound names $\bnames{\mu}$ of $\mu$
are defined as
$\bnames{\out a b}≜\bnames{a b}≜\bnames{\tau}≜\emptyset$ and
$\bnames{\out a(b)}≜\{b\}$.
%

\begin{figure}[h]
  \begin{mathpar}
    \newcommand*{\aerate}{\phantom{b\!\!}}
    \inferrule*[lab=out]{ }{ \out a b.P \tpi{\out a b} P } \and
    \inferrule*[lab=inp]{ }{ a(b).P \tpi{a c\aerate} P\{c/b\} } \\
    \inferrule*[lab=open,right={$a≠b,c\not\in\fnames{\res b P}$}]{P\tpi{\out a b}P'}{ \res b P \tpi{\out a(c)} P'\{c/b\} } \and
    \inferrule*[lab=res,right={$a∉\names{μ}$}]{P\tpi{μ\aerate}P'}{ \res a P \tpi{μ\aerate} \res a P' } \and
    \inferrule*[lab=comm-l]{P\tpi{\out a b}P' \and Q \tpi{a b} Q'}{ P∣Q \tpi{τ} P'∣Q' } \and
    \inferrule*[lab=close-l,right={$b∉\fnames{Q}$}]{P\tpi{\out a(b)}P' \and Q \tpi{a b} Q'}{ P∣Q \tpi{τ} \res b (P'∣Q') } \\
    \inferrule*[lab=sum-l]{P\tpi{μ}P'}{ P+Q \tpi{μ} P' } \and
    \inferrule*[lab=par-l,right={$\bnames{μ}∩\fnames{Q}=∅$}]{P\tpi{μ}P'}{ P∣Q \tpi{μ} P'∣Q } \and
    \inferrule*[lab=rep]{P ∣ {!}P\tpi{μ}P'}{ {!}P \tpi{μ} P' } \and
  \end{mathpar}
  \caption[The $\pi$-calculus]{Operational semantics of the
    $\pi$-calculus \\(symmetric \textsc{-r} versions of
    \textsc{-l} rules are omitted)}
  \label{f:pi}
\end{figure}


Note that the conclusion of the \textsc{open} rule is instead
$\res b P \tpi{\out a(b)} P'$ in some presentations of the
$\pi$-calculus. Those presentations look simpler but rely on
$\alpha$-conversion of $b$ in $\res b P$. We choose here to be more
explicit.

We do not want to distinguish processes according to the identity of
the bound names they may extrude. This is why we need a specific
clause for bound outputs in the standard definition of bisimulation:
\begin{definition}
  \label{def:early:pi}
  A relation $\RR$ is a strong early bisimulation if, whenever $P \R Q$:
  \begin{enumerate}
  \item \label{bisimclause1} if $\smash{P\tpi{\out a(b)}P'}$ and
    $b∉\fnames{Q}$ then $\smash{Q\tpi{\out a(b)}Q'}$ for some $Q'$
    such that $P' \R Q'$,
  \item \label{bisimclause2}if $P\tpi{μ}P'$ and $μ$ is not a bound output, then
    $Q\tpi{μ}Q'$ for some $Q'$ such that $P' \R Q'$,
  \item the converse of (\ref{bisimclause1}) and (\ref{bisimclause2}), on $Q$.
  \end{enumerate}
  \emph{Early bisimilarity}, $\simE$, is the union of all early
  bisimulations. The weak version of early bisimilarity, \emph{weak
    early bisimilarity}, written $\wbE$, is obtained in the standard
  way: the transition ${Q\tpi{\out a(b)}Q'}$ in clause
  (\ref{bisimclause1}) is replaced by ${Q\wpi{\out a(b)}Q'}$; and
  similarly the transition $Q\tpi{μ}Q'$ in (\ref{bisimclause2}) is
  replaced by $Q \wpi{\hat\mu} Q'$.  The $\wpi{\cdot}$
  transitions are defined from $\tpi{\cdot}$ the same way the
  $\ws{\cdot}$ transitions were from $\ts{\cdot}$.
\end{definition}
When translating the $\pi$-calculus semantics to a first-order one,
the ad-hoc condition $b∉\fnames{Q}$ has to be removed.  To this end,
one has to force an agreement between two bisimilar processes on the
choice of the bound names appearing in transitions.
We obtain this by considering \emph{named processes} $(c,P)$ in which
$c$ is bigger or equal to all names in $P$.  For this to make sense we
assume an enumeration of the names and use $≤$ as the underlying
order, and $c+1$ for name following $c$ in the enumeration; for a set
of names $N$, we also write $c \geq N$ to mean $c \geq a$ for all
$a\in N$.

The rules below define the translation of the $\pi$-calculus
transition system to a first-order LTS.  In the first-order LTS, the
grammar for labels is the same as that of the original LTS; however,
for a named process $(c,P)$ the only name that may be exported in a
bound output is $c+1$; similarly only names that are below or equal to
$c+1$ may be imported in an input transition.  (Indeed, testing for
all fresh names $b>c$ is unnecessary, doing it only for one ($b=c+1$)
is enough.)  This makes it possible to use the ordinary definition of
bisimilarity for first-order LTSs, and thus recover the early
bisimilarity on the source terms.
\begin{mathpar}
  \inferrule*{P \tpi{τ} P'}{(c,P) \ts{τ} (c,P')} \and
  \inferrule*[right=$b≤c$]{P \tpi{a b} P'}{(c,P) \ts{a b} (c,P')} \and
  \inferrule*[right=$b≤c$]{P \tpi{\out a b} P'}{(c,P) \ts{\out a b} (c,P')} \\
  \inferrule*[right=${b=c+1}$]{P \tpi{a b} P'}{(c,P) \ts{a b} (b,P')} \and
  \inferrule*[right=${b=c+1}$]{P \tpi{\out a(b)} P'}{(c,P) \ts{\out a(b)} (b,P')} 
\end{mathpar}

We write $\pi^1$ for the first-order LTS derived from the above
translation of the $\pi$-calculus. Although the labels of the source
and target transitions have a similar shape, the LTS in $\pi^1$ is
first-order because labels are taken as purely syntactic,
uninterpreted objects.  We can also define $\pi^1$ using the
following two rules:
\begin{mathpar}
  \inferrule*{P \tpi{\mu} P'\and \names{\mu}\leq c\and{\bnames{\mu}=\emptyset}}{(c,P) \ts{\mu} (c,P')} \and
  \inferrule*{P \tpi{\mu} P'\and c+1\in\names{\mu}}{(c,P) \ts{\mu} (c+1,P')} \and
\end{mathpar}
This characterisation is less explicit but sometimes more convenient
in proofs, and it might give an insight on to how to derive
translations for other name-based calculi by keeping track of new
names and of binding labels.

We will show that the standard notions strong and weak early
bisimilarity for the $\pi$-calculus ($\simE$ and $\wbE$ from
Definition~\ref{def:early:pi}) correspond to $\sim$ and $\wb$ in
$\pi^1$. Proving soundness, i.e., bisimilarity in $\pi^1$ entails
bisimilarity in $\pi$, requires us to establish first that
bisimilarity in $\pi^1$ is stable under injective substitutions.
Anticipating that we also want to propose various up-to techniques for
$\pi^1$, we show directly that the corresponding
up-to-injective-substitutions technique is below the companion. It
follows that bisimilarity in $\pi^1$ is stable under injective
substitutions by Lemma~\ref{l:compatible:closure}, and the work is
done only once.

We define the following monotone functions on relations on $\pi^1$
processes:
\[\begin{array}{rll}
  \isubst(\R) &\defi \{ ((d,P σ), (d,Q σ)) \; \; &\st\;\;
                (c,P) \RR (c,Q),~\text{$σ$ injective on }\fnames{P} \cup \fnames{Q},\\
              &&\hspace{4cm}\text{and }\fnames {P σ} \cup \fnames {Q σ}≤d\}\\
  \bsubst(\R) &\defi \{ ((c,P σ), (c,Q σ)) \; \; &\st\;\;
                (c,P) \RR (c,Q),~\text{$σ$ bijective on }\{1\dots c\} \}\\
  \str(\R) &\defi  \{(  (d,P) ,  (d,Q)) ~\,&\st\;\; (c,P) \R
             (c,Q)\text{ and }\fnames{P,Q}≤d   \}\\
  \weak(\R)&\defi  \{((c+k,P), (c+k,Q)) \;\;&\st\;\; (c,P) \R (c,Q) ~,~ k ∈ ℕ \}
  \end{array}\]
The first one, \isubst, makes it possible to use injective
substitutions; the second one, \bsubst, is restricted to bijective
substitutions; the third one, \str, is a form of \emph{strengthening},
making it possible to readjust the bound $c$ on free names;
conversely, the last one, $\weak$, is a form of \emph{weakening}. The last two
functions are often useful as up-to techniques, by themselves.  The
point of the function \bsubst\ is that it makes it possible to obtain
\isubst\ as a derived technique: we have $\isubst = \bsubst ∘ \weak$, and
$\bsubst$ is slightly easier to analyse.

\begin{lemma}
  \label{l:isubstr}
  The functions $\isubst$, $\bsubst$, $\str$, and $\weak$ are all below
  the companion $t_{\strongb}$ and below the companion $t_{\weakb}$.
\end{lemma}
\begin{proof}
  We first show that $\bsubst$ is compatible, i.e.,
  $\bsubst \funprogress{} \bsubst$. Let $\RR,{\SS}$ be two relations
  such that ${\RR}\relprogress{}{\SS}$, and let us prove
  $\bsubst(\RR)\relprogress{}\bsubst({\SS})$. For this, let
  $(Pσ,Qσ) \in \bsubst(\RR)$ for some $P,Q$ such that $(c,P)\R (c,Q)$
  with $σ$ some bijective substitution on $\{1\dots c\}$. From a transition
  $(c,Pσ) \ts{\mu} (c',P')$, because $Pσσ^{-1}=P$, we can transform it
  to $(c,P) \ts{\mu σ^{-1}} (c', P'σ^{-1})$. Then, thanks to
  $\RR\relprogress{}{\SS}$ with the $\mu σ^{-1}$ transition, there
  exists $Q'$ such that $(c,Q) \ts{\mu σ^{-1}} (c',Q')$ and
  $(c',P'σ^{-1}) \S (c',Q')$, which imply respectively
  $(c,Qσ) \ts{\mu} (c',Q'σ)$ and
  $(c',P'σ^{-1}σ) = (c',P') \mathrel{\bsubst({\SS})} (c',Q'σ)$. The
  argument used for $\strongb$ can also be used for $\weakb$.
  Thus $\bsubst$ is below $t$.
  
  Then we show $\weak \funprogress{} \bsubst ∘ \weak$ and
  $\str \funprogress{} \str ∘ \bsubst$, which is done using a similar
  diagram-chasing argument. Each newly created name is handled with a
  transposition using $\bsubst$, using the facts that $Q\tpi{\mu}Q'$
  implies $\fnames{Q'} \subseteq \fnames{Q} \cup \names{\mu}$, and
  that $\fnames{Q'\sigma} = \sigma(\fnames{Q'})$.

  Since $\bsubst\subseteq t$, we deduce $\weak \funprogress{} T(\weak)$ and
  $\str \funprogress{} T(\str)$, so that both $\weak$ and $\str$ are also
  below $t$. It follows that
  $\isubst = \bsubst ∘ \weak \subseteq t \circ t = t$.
\end{proof}
It follows by Lemma~\ref{l:compatible:closure} that bisimilarities in
$\pi^1$ are closed under injective substitution:
$\isubst(\sim)\subseteq{\sim}$ and
$\isubst(\approx)\subseteq{\approx}$.
We can now establish full abstraction between $\pi^1$ and early
bisimilarities:
\begin{theorem}
  \label{t:pi:corr} Assume $c≥\fnames{P} ∪ \fnames{Q}$. Then we have:
  $P \simE Q$ iff $(c,P)\sim (c,Q)$, and $P \wbE Q$ iff $(c,P)≈
  (c,Q)$.
\end{theorem}
\begin{proof}
  We prove the case of weak bisimilarity, the strong case being
  easier. For the direct implication, we show that the relation
  $\RR_1$ defined below is a weak bisimulation:
  \[{\RR}_1\eqdef \{((c,P), (c,Q)) \mid P \wbE Q ~∧~ c≥\fnames{P} ∪
    \fnames{Q}\} \]%
  The only interesting transition is when
  $(c,P) \ts{\out a(b)} (d,P') $ with $d=b=c+1$.
  Since $c ≥\fnames{P} ∪ \fnames{Q}$, we know that $b ∉\fnames{Q}$ so
  $P \wbE Q$ tells us that $Q\wpi{\out a(b)}Q'$ with $P' \wbE Q'$.
  Repeatedly applying the rules defining $\ts{\cdot}$, since $b=c+1$, yields
  \[
  \inferrule{
    Q\wpi{}Q_1\tpi{\out a(b)}Q_2\wpi{} Q'}{
    (c,Q)\ws{}(c,Q_1)\ts{\out a(b)}(b,Q_2)\ws{} (b,Q')
  }
  \]
  and, indeed, $b≥\fnames{P'} ∪ \fnames{Q'}$.
 \medskip
  
  For the converse, proving that
  \[ {\RR}_2 \eqdef \{(P,Q) \mid ∃ c≥\fnames{P} ∪ \fnames{Q} ~~ (c,P) \wb
    (c,Q)\} \] is a weak early bisimulation needs a little more care,
  since fresh names in labels can be other than $c+1$ (they can be less
  than or greater than $c+1$). Suppose $P \R_2 Q$, which means there
  is  $c≥\fnames{P} ∪ \fnames{Q}$ such that $(c,P) \wb (c,Q)$. We
  analyse the transitions of the form $P\tpi{μ}P'$:
  \begin{enumerate}
  \item if $μ=\out a(b)$ or $μ = a b$ and $b\geq c+2$ then we have
    $P\tpi{μ}P'$ if and only if $P\tpi{μ\subst{b'}{b}}P'\subst{b'}{b}$
    with $b'=c+1$. Exploiting $\approx$ to get
    $(c,Q)\ws{\out a(b')}(b',Q')$ with
    $(b',P'\subst{b'}{b})\approx (b',Q')$ and hence, since
    $\isubst(\approx)\subseteq {\approx}$ by
    Lemmas~\ref{l:compatible:closure} and~\ref{l:isubstr}, we obtain
    $(b,P')\approx (b,Q'\subst{b}{b'})$ and hence
    $P'\R_2 Q'\subst{b}{b'}$. We conclude since
    $Q\wpi{μ\subst{b'}{b}}Q'$ implies $Q\wpi{μ}Q'\subst{b}{b'}$.
  \item If $μ = a b$ or $μ = \out a(b)$ with $b=c+1$ then $P\tpi{μ}P'$
    implies $(c,P)\ts{μ}(b,P')$, which implies $(c,Q)\ws{μ}(b,Q')$, and
    then $Q\wpi{μ}Q'$ with $(b,P') \wb (b,Q')$.
  \item If $\names{μ}\leq c$, there are two cases:
    \begin{enumerate}
    \item $μ$ is not a bound output. Then $P\tpi{μ}P'$ if and only if
      $(c,P)\ts{μ}(c,P')$ when $c\geq \fnames{P}$;  thus we can derive
      $Q\wpi{μ}Q'$ and $P' \R_2 Q'$.
    \item $μ =\out a(b)$ (hence $b∉\fnames{P}$) with the additional
      information that $b∉\fnames{Q}$. Then
      $P \tpi{\out a(b')} P'\subst{b'}{b}$ with $b'=c+1$. We get
      $(c,P) \ts{\out a(b')} (b',P'\subst{b'}{b})$ and then
      $(c,Q) \ws{\out a(b')} (b',Q')$ and since $b∉\fnames{Q}$, we
      also have $Q \wpi{\out a(b)} Q'\subst{b}{b'}$. The progression
      also gives us $(b',P'\subst{b'}{b}) \wb (b',Q')$. By closure of
      $\approx$ under injective substitutions (\isubst), we deduce
      $(b,(P'\subst{b'}{b})\subst{b}{b'}) = (b,P') \wb
      (b,Q'\subst{b}{b'})$ and finally $P' \R_2 Q'\subst{b}{b'}$.
      \qedhere
    \end{enumerate}
  \end{enumerate}
\end{proof}

The above full abstraction result allows us to import the theory of
up-to techniques for first-order LTSs and bisimilarity, in both the
strong and the weak cases.

We have already proved the validity of preliminary up-to techniques
that are specific to $\pi^1$ (Lemma~\ref{l:isubstr}); we proceed 
below with up-to-context techniques.

The up-to-context function is decomposed into a set of smaller context
functions, called \emph{initial}~\cite{SanPous}, one for each operator
of the $\pi$-calculus.  The only exception to this is the input
prefix, since early bisimilarity in the $\pi$-calculus is not
preserved by this operator.  We write $\icout$, $\icnew$, $\icbang$,
$\icpar$, and $\icplus$ for these initial context functions,
respectively applying the operators of output prefix, restriction,
replication, parallel composition, and sum, to all pairs in the given
relation.

\begin{definition}\label{d:pi:uptocontext}
  We define the functions $\icout$,
  $\icnew$, $\icbang$, $\icpar$ and $\icplus$ on relations on $\pi^1$ by the following rules:
  \begin{mathpar}
  \inferrule{(c,P)\R(c,Q) \and c\geq a,b}{(c,\out a b.P) ~~\icout(\RR)~~ (c,\out a b.Q)} \and
  \inferrule{(c,P)\R(c,Q)}{(c,(\new a)P) ~~\icnew(\RR)~~ (c,(\new a)Q)} \and
  \inferrule{(c,P)\R(c,Q)}{(c,{!}P) ~~\icbang(\RR)~~ (c,{!}Q)} \and
  \inferrule{(c,P_1)\R(c,Q_1) \and (c,P_2)\R(c,Q_2)}{(c,P_1∣P_2) ~~\icpar(\RR)~~ (c,Q_1∣Q_2)} \quad
  \inferrule{(c,P_1)\R(c,Q_1) \and (c,P_2)\R(c,Q_2)}{(c,P_1+P_2) ~~\icplus(\RR)~~ (c,Q_1+Q_2)} \and
  \end{mathpar}
\end{definition}

While bisimilarity in the $\pi$-calculus is not preserved by input
prefix, a weaker rule holds:
\begin{equation}
  \label{e:eIP}
  \displaystyle{\frac{
      \forall c, \;\; P \sub c b \asymp Q \sub c b
    }
    {a(b).P  \asymp a(b).Q}
  }
\end{equation}
where $\asymp$ can be $\simE$ or $\wbE$. We define accordingly $\icin$, the
function for input prefix:
\begin{definition}$\icin$ is the function on $\pi^1$ relations defined
  by the rule:
  \[\inferrule
    {a\leq d \and \forall c\leq d+1 ~~~(d+1,P\{c/b\})\R(d+1,Q\{c/b\})}
    {(d,a(b).P) ~~\icin(\RR)~~ (d, a(b).Q))}\]
\end{definition}

\begin{theorem}
  \label{t:pi:uptoc}
  The functions $\icout, \icin, \icnew, \icbang, \icpar, \icplus$ are
  all below $t_{\strongb}$.
\end{theorem}
\begin{proof}
  Let $F \eqdef \{\icout, \icin, \icnew, \icbang, \icpar, \icplus\}$,
  we prove $f \funprogress{} T(\cup F)$ for each function $f\in F$. In
  each case, we assume ${\RR}\relprogress{} {\SS}$ and we prove
  $f({\RR}) \relprogress{} T(\cup F)({\SS})$. Remark that
  $𝓡⊆\strongb(𝓢)$, and so
  $T(\cup F)(𝓡∪𝓢) ⊆ T(\cup F)(\strongb(𝓢)∪𝓢)⊆ T(\cup F)(T(\cup F)(𝓢))
  =T(\cup F)(𝓢)$, and so it is enough, and more convenient, to
  establish $f({\RR}) \relprogress{} T(\cup F)(𝓡∪𝓢)$.
  
  For this, it suffices to
  analyse the transitions emerging from the left-hand side of
  $f(\RR)$, as every $f$ is symmetric:
  \begin{itemize}
  \item $\icout(\RR) \relprogress{} {\RR}$ is immediate
  \item $\icplus(\RR) \relprogress{} {\SS}$ is also
    straightforward;
  \item $\icin(\RR) \relprogress{} \str(\RR)$:
    assume $(d,a(b).P) ~~\icin(\RR)~~ (d,a(b).Q)$; each transition is of
    the form $\ts{a c}$, yielding a pair $(p,q)$ where
    $p=(d',P\{c/b\})$ and $q=(d',Q\{c/b\})$;
    \begin{itemize}
    \item if $d'=c+1$ then $(p,q)∈{\RR}$ by definition of $\icin$.
    \item if $d'=c$ then $(d+1,P\{c/b\}) \R (d+1,Q\{c/b\})$ by
      definition of $\icin$, and hence
      $(p,q) ∈ \str(\RR)$.
    \end{itemize}
  \item $\icnew(\RR) \relprogress{} \icnew(\SS) ∪ \isubst(\SS)$\\
    The interesting case arises for transitions for which the last
    rule applied is the extrusion rule: $(c,(\new d)P)$
    $\ts{\out a(b)}$ $(b,P'\subst{b}{d})$ with $b=c+1$ and
    $P \tpi{\out a d}P'$. The problem is to relate
    $(b,P'\subst{b}{d})$ to $(b,Q'\subst{b}{d})$ knowing that
    $(c,P') \S (c,Q')$. This is done using the $\isubst$ function
    with the injective substitution $\subst{b}{d}:\{1\dots c\}→\{1\dots b\}$.
  \item $\icpar(\RR) \relprogress{} N(\icpar(\isubst({\RR}\cup{\SS})))$, where
    $N \eqdef (\str ∘ \icnew) ∪ \id$. For this, we analyse the transition
    $(c, P_1 \mid P_2) \ts{\mu}(c', P')$.
        First, let's assume that $c'=c$. The transition must come from one
    of the four rules \textsc{par-l}, \textsc{par-r}, \textsc{comm-l},
    or \textsc{comm-r}:
    \begin{itemize}
    \item Rule \textsc{par-l} results in $(c, P_1' \mid P_2)$ with
      $P_1 \tpi{\mu_1} P_1'$, so we obtain, from
      $(c, P_1) \RR (c, Q_1)$ and ${\RR}\relprogress{} {\SS}$, some $Q_1'$
      such that $Q_1 \tpi{\mu_1} Q_1'$. Finally we obtain the pair
      $((c,P_1'\mid P_2),(c,Q_1'\mid Q_2))$ which belongs to
      $\icpar({\RR}∪{\SS})$.
    \item Symmetrically, rule \textsc{par-r} takes us to the pair
      $((c,P_1\mid P_2'),(c,Q_1\mid Q_2'))\in\icpar({\RR}\cup{\SS})$.
    \item Working both sides, rules \textsc{comm-l} and \textsc{comm-r} both
      lead us to a pair
      $$((c,P_1'\mid P_2'),(c,Q_1'\mid Q_2'))\in\icpar(\SS).$$
    \end{itemize}
    ~
    
    The second case is when $c'=c+1$. This means that the transition is
    derived using \textsc{par-l}, \textsc{par-r}, \textsc{close-l}, or
    \textsc{close-r}. We consider two cases:
    \begin{itemize}
    \item The last rule is a \textsc{par-l} rule
      (\textsc{par-r} being symmetric), with a label of the form
      $\out a(b)$. We know $(c,P_1)\R(c,Q_1)$ and $(c,P_2)\R(c,Q_2)$
      and
      \[
        \inferrule*[right={$b∉\fnames{P_2}$}]{P_1\tpi{\out a(b)}P_1'}{
          P_1∣P_2 \tpi{\out a(b)} P_1'∣P_2 } \] We have $b=c+1$,
      following the rule for bound output. We also have the following
      reductions in $\pi^1$, from $(c,P_1)$, and then from $(c,Q_1)$ using
      the progression ${\RR}\relprogress{} {\SS}$:
      \begin{mathpar}
        (c, P_1)\ts{\out a(b)}(b, P_1') \and
        (c, Q_1)\ts{\out a(b)}(b, Q_1') \and
      \end{mathpar}
      With $(b, P_1') \S (b, Q_1')$. We now need to relate the
      resulting processes:
      \[
      \inferrule{
        \inferrule{
          \inferrule
          {(b, P_1')\S (b, Q_1')}
          {(b, P_1')\R∪\S (b, Q_1')}
        }{(b, P_1') ~\isubst(\R∪\S)~ (b, Q_1')}
      \and
      \inferrule{
        \inferrule{(c, P_2 )\R (c, Q_2 )}{(c, P_2) \R∪\S (c, Q_2)}}
        {(b, P_2)~\isubst(\R∪\S)~(b, Q_2)}}
      {(b, P_1' ∣ P_2)~(\icpar∘\isubst)(\R∪\S)~(b, Q_1' ∣ Q_2)}
      \enspace.\]
      The same happens for the input transition $\ts{a b}$ when $b=c+1$.
    \item The last rule is a \textsc{close} rule: we know $(c,P_1)\R(c,Q_1)$
      and $(c,P_2)\R(c,Q_2)$ and
      \[
        \inferrule*[right={$b∉\fnames{P_2}$}]{P_1\tpi{\out a(b)}P_1'
          \and P_2 \tpi{a b} P_2'}{ P_1∣P_2 \tpi{τ} (\res
          b)(P_1'∣P_2') } \] 
      We can assume $b=c+1$ as $b$ is
      fresh on both sides. The two hypotheses can then be transformed
      into transitions in $\pi_1$:
      \begin{mathpar}
        (c, P_1)\ts{\out a(b)}(b, P_1') \and
        (c, P_2) \ts{a b} (b, P_2')\enspace.
      \end{mathpar}
      We have the same transitions for $Q_1$ and $Q_2$, respectively.
      Using the hypothesis ${\RR}\relprogress{} {\SS}$, we obtain named
      processes $(b,Q_1')$ and $(b,Q_2')$, related through $\SS$, which
      we can combine using $\icpar$ and then strengthen $b$ to $c$
      since $b\not\in\fnames{P,Q}$:
      \begin{mathpar}
        \inferrule{
          \inferrule{
            \inferrule{(b, P_1') \S (b, Q_1') \qquad\qquad (b, P_2') \S (b, Q_2')}
            {(b, P_1' ∣ P_2')\qquad \icpar(\S)\qquad (b, Q_2' ∣ Q_1')}
          }{
            (b, (\new b)(P_1' ∣ P_2'))\quad \icnew(\icpar(\S))\quad 
            (b, (\new b)(Q_2' ∣ Q_1'))
          }}{
          (c, (\new b)(P_1' ∣ P_2'))~~\str(\icnew(\icpar(\S)))~~
          (c, (\new b)(Q_2' ∣ Q_1'))
        } \enspace.
      \end{mathpar}
    \end{itemize}
  \item
    $\icbang(\RR) \relprogress{} \icpar^\omega(N(\icpar^\omega((\icbang ∪ \id)(\isubst(\R\cup\S)))))$\\
    We analyse the transition $(c,{!}P)\ts{μ}(c',P')$. 
    If
    ${!}P\tpi{μ}P'$ then one of the following holds:
    \newcommand{\B}{\!∣\!}
    \begin{enumerate}
    \item $P' = {!}P \B P_0 \B P \B … \B P$ with $P\tpi{μ}P_0$, or 
    \item $μ=τ$ and
      $P' = {!}P \B P_0 \B P \B … \B P \B P_1 \B P \B … \B P$ with
      $P\tpi{\out a b}P_i$ and $P\tpi{a b}P_{1-i}$, or
    \item $μ=τ$ and
      $P' = (\new b)({!}P \B P_0 \B P \B … \B P \B P_1) \B P \B … \B
      P$ with $P\tpi{\out a(b)}P_i$ and $P\tpi{a b}P_{1-i}$.
    \end{enumerate}
    In case 1 we have $(c,P)\ts{μ}(b,P_0)$ with $b=c'∈\{c,c+1\}$.
    In case 2 we have $(c,P)\ts{μ_i}(b,P_i)$ for each $i$, with $b=c$.
    In case 3 we have $(c,P)\ts{μ_i}(b,P_i)$ for each $i$, with $b=c+1$.
    In each case, we obtain, since ${\RR}\relprogress{} {\SS}$ with
    $(c,P)\R(c,Q)$, a transition $(c,{!}Q)\ts{μ}(c',Q')$ with $Q'$ of
    the same shape, so we only need to relate $(c',P')$ to $(c',Q')$
    knowing $(b,P_i)\S(b,Q_i)$. First,  we note that
    $(b,P)~\isubst(\R)~(b,Q)$. We have now the following pairs in
    ${\SS_0} ≜ (\icbang ∪ \id)(\isubst(\R\cup\S))$:
    \begin{mathpar}
      (b,{!}P)\SS_0(b,{!}Q) \and
      (b,P_0)\SS_0(b,Q_0) \and
      (b,P_1)\SS_0(b,Q_1) \and
      (b,P)\SS_0(b,Q) \and
    \end{mathpar}
    We can then  apply $\icpar$ several times to obtain the three pairs
    (with $\SS_1≜\icpar^\omega(\SS_0)$):
    \begin{align*}
      (b, {!}P \B P_0 \B P \B … \B P) ~~~~ &\SS_1 ~~~~
      (b, {!}Q \B Q_0 \B Q \B … \B Q)\\
      (b, {!}P \B P_0 \B P \B … \B P \B P_1 \B P \B … \B P) ~~~~ &\SS_1 ~~~~
      (b, {!}Q \B Q_0 \B Q \B … \B Q \B Q_1 \B P \B … \B P)\\
      (b, {!}P \B P_0 \B P \B … \B P \B P_1) ~~~~ &\SS_1 ~~~~
      (b, {!}Q \B Q_0 \B Q \B … \B Q \B Q_1)
    \end{align*}
    The first two pairs handle cases 1 and 2. For case 3 we need to
    apply $\icnew$ to add $(\new b)-$ and then $\str$ so to go  from
    $(b,(\new b)-)$ to $(c,(\new b)-)$. We apply $\icpar^\omega$ again
    to add the missing $- \B P \B … \B P$ and we obtain $(c,P')$ and
    $(c,Q')$ in the relation $\icpar^\omega(\str(\icnew(\SS_1)))$.
    Concluding, we have  obtained the following progression:
    \[
      \icbang({\RR}) \relprogress{} {\SS}_1 ∪ \icpar^\omega(\str(\icnew(\SS_1))) ⊆
      \icpar^\omega(N(\SS_1))
    \]
    which was our original goal. Note that the iterated
    $\icpar^\omega$ was used twice; both times it can be absorbed by
    $T$,  so to give us at the end
    $\icbang(\RR) \relprogress{} T(\icpar \cup \icnew \cup
    \icbang)(\R\cup\S)$.
  \end{itemize}
  For each $f\in F$ we have established a progression from $f(\RR)$ to
  $T(∪F)(\R\cup\S)$, and so to $T(∪F)(\S)$, as needed: this gives us
  $f\funprogress{} T(∪F)$, and in turn $∪F \funprogress{} T(∪F)$ and
  $∪F \subseteq t$.
\end{proof}

Weak bisimilarity is not preserved by sums, only by guarded sums,
whose function is $\icgplus \defi \icplus^ω ∘ (\icout ∪ \icin)$.

\begin{theorem}
  \label{t:pi:uptoc:weak}
  The functions $\icout, \icin, \icnew, \icbang, \icpar, \icgplus$ are
  below the companion $t_{\weakb}$.
\end{theorem}
\begin{proof}
  The progressions are as in the proof of Theorem~\ref{t:pi:uptoc},
  except for $\icgplus$, which is treated as $\icout$ and $\icin$; we
  need one more up-to technique for the case of the replication.
  Assuming ${\RR} \relprogress{\weakb} {\SS}$, the following
  progressions hold:
  \begin{mathpar}
    \icout(\RR)  \relprogress{\weakb} {\RR} \and
    \icin(\RR)   \relprogress{\weakb} \str(\RR) \and
    \icnew(\RR)  \relprogress{\weakb} (\icnew ∪ \isubst)(\SS) \and
    \icgplus(\RR)\relprogress{\weakb} \str(\RR) \and
    \icpar(\RR)  \relprogress{\weakb} N(\icpar(\isubst({\RR} \cup {\SS}))) \and
    \icbang(\RR) \relprogress{\weakb} \icpar^\omega(N(\icpar^\omega((\icbang ∪ \id)(\isubst({\RR}\cup{\SS})))))
    \cup (\icpar^\omega(\icbang({\RR})∪{\RR}∪{\SS}){\sim})
    \enspace.
  \end{mathpar}
  Again progressing to $T(∪F)(𝓡∪𝓢)$ is conveniently sufficient since
  $T(∪F)(𝓡∪𝓢)⊆ T(∪F)(𝓢)$.
  \newcommand{\B}{\!∣\!}%
  For the replication operator, only case 1 (of the corresponding proof
  of Theorem~\ref{t:pi:uptoc}) cannot be transported to the weak case. We have:
  \begin{mathpar}
  (c,{!}P) \ts{τ} (c,{!}P\B P_0\B P \B … \B P) \and\mbox{with}\and
  (c,P) \ts{τ} (c,P_0)\enspace.
  \end{mathpar}
  We use the property that ${\RR} \relprogress{} {\SS}$ so that from
  $(c,P)\R (c,Q)$ we obtain $(c,Q) \ts{τ}^n (c,Q_0)$. Then
  \begin{itemize}
  \item if $n>0$ we have $(c,{!}Q) \ws{} (c,{!}Q \B Q_0 \B …\B Q)$ and
    we conclude as before.
  \item if $n=0$ then there is no transition from $Q$ or $!Q$, we know
    $P_0 \S Q$ but we cannot reach the desired form
    $(c,{!}Q \B Q \B …\B Q)$ with a transition. Instead, we remark
    that $(c,{!}Q) \sim (c,{!}Q \B Q \B …\B Q)$ and so we simply
    progress to the relation $\S_2\sim$ where
    ${\SS}_2 = \icpar^\omega(\icbang({\RR})∪{\RR}∪{\SS})$.
  \end{itemize}
  Compared to the strong case, we only need to compose (on the left)
  the right-hand side of the progression with
  the function ${\RR}↦{\sim\R\sim}$  (`up-to-strong-bisimilarity')  which is
  indeed $\weakb$-compatible.
\end{proof}

As a byproduct of the compatibility of these initial context
functions, and using Lemma~\ref{l:compatible:closure}, we derive the
standard congruence properties of strong and weak early
bisimilarity, including the rule (\ref{e:eIP}) for input prefix.

\begin{corollary}
  In the $\pi$-calculus, relations $\simE$ and $\wbE$ are preserved by
  the operators of output prefix, replication, parallel composition,
  restriction; $\simE$ is also preserved by sum, whereas $\wbE$ is
  only preserved by guarded sums.  Moreover, rule (\ref{e:eIP}), for input prefix,  is
  valid both for $\simE$ and~$\wbE$.
\end{corollary}

\begin{remark}
  \label{rem:early:late}
  Late bisimilarity~\cite[Section 4.5]{SW01} makes use of transitions $P \arrpi{a(b)} P'$ where
  $b$ is bound, the definition of bisimulation containing a
  quantification over names. To translate this bisimilarity in a
  first-order LTS we would need two transitions for the input $a(b)$:
  one to fire the input~$a$, leaving $b$ uninstantiated (for example,
  in a new kind of process $(b)(c,P)$ akin to an abstraction), and
  another to instantiate $b$ with any name, for transitions starting
  from processes of the new kind:
  \begin{mathpar}
    \inferrule{P\arrpi{a(b)}P'}{(c,P)\ts{a(-)}(b)(c,P')} \and
    \inferrule{b'\leq c+1 }{(b)(c,P)\ts{a(b')}(c+1,P'\subst{b'}{b})}
  \end{mathpar}
  While such a translation does yield full abstraction for both strong
  and weak late bisimilarities, the decomposition of an input
  transition into two steps prevents us from obtaining the
  compatibility of up-to-context. Indeed, compatibility of
  up-to-context intuitively requires that the immediate transitions of
  $C[P]$ should depend only on the immediate transitions of $P$.
  However, if inputs are decomposed into two steps, a contexts such as
  $[\cdot]_1\mid\out a b$ may combine two successive steps of the
  (input) argument to perform a single $\tau$ transition.
\end{remark}

To conclude, the main take-away message on the $\pi$-calculus is that
it suffices to count names to make the LTS first-order. Then, once the
corresponding up-to techniques for names are set-up, we recover the usual
progression proofs, in a modular way. While this level of modularity was
already present in \cite{phd:pous}, it  now becomes simpler
thanks to the companion.


\section{Call-by-name λ-calculus}
\label{s:cbn}

To study the applicability of our approach to higher-order languages,
we investigate the pure call-by-name $\lambda$-calculus, referred to
as $\LaN$ in the sequel.

We use $M,N$ to range over the set $\Lambda$ of $\lambda$-terms, and
$x,y,z$ to range over variables. The set $\Lambda$ of pure
$\lambda$-terms is defined by:
\[ M, N ::= x \mid \lambda x.M \mid M N\] We assume the familiar
concepts of free and bound variables and substitutions, and identify
$\alpha$-convertible terms. The only values are the
$\lambda$-abstractions $\lambda x. M$. In this section and in the
following one, results and definitions are presented on closed terms
and we write $\Lao$ for the subset of closed terms. Extension to open
terms is made using closing abstractions (i.e., abstracting on all
free variables). The reduction relation of $\LaN$ is the
\emph{call-by-name reduction relation} $\cbn$, defined as the least
relation over $\Lao$ that is closed under the following rules.
\[ 
\inferrule{ }{(\lambda x. M) N \cbn M\sub N x} \qquad\qquad
\inferrule{M \cbn M' } { M N \cbn M' N}
\] 
We write $\wcbn$ for its reflexive and transitive closure. In
call-by-name, \emph{evaluation contexts} are described by the following
grammar:
\[C_e :=  C_e\; M   \mid \holE   \]

As reference equivalence for the $\lambda$-calculus we consider
\emph{environmental bisimilarity} \cite{envbisim,KoutavasLS11}, which
coincides with contextual equivalence and Abramsky's applicative
bisimilarity~\cite{Abr88} on pure $\lambda$-terms while enabling a richer
set of up-to techniques.  Environmental bisimilarity makes a clear
distinction between the tested terms and the environment.  An element
of an environmental bisimulation has, in addition to the tested terms
$M$ and $N$, a further component $\EE$, the environment, which
expresses the observer's current knowledge.  When an input from the
observer is required, the arguments supplied are terms that the
observer can build using the current knowledge; that is, terms
obtained by composing the values in $\EE$ using the operators of the
calculus. An \emph{environmental relation} is a set of elements, each
of which can be of two forms: either a relation $\EE$ on closed
values, or a triple $(\EE, M,N )$ where $M,N$ are closed terms and
$\EE $ is a relation on closed values. We use $\X,\Y$ to range over
environmental relations. In a triple
$(\EE,M,N)$ the relation component $\EE$ is the \emph{environment},
and $M,N$ are the \emph{tested terms}.  We write $M \Xv\EE N$ for
$(\EE, M, N) \in \X$.  We write $\starred\EE$ for the closure of $\EE$
under contexts.  We only define the weak version of the bisimilarity;
its strong version is obtained in the expected way.

\begin{definition}
  \label{d:bisim-for-lambdaN} 
  An environmental relation $\X$\! is an \emph{environmental
    bisimulation} if
  \begin{enumerate}
  \item $M \Xv\EE N$ implies:
    \begin{enumerate}
    \item if $M \cbn M'$ then $N \wcbn N'$ and $M' \Xv{ \EE} N'$;
    \item if $M=V$ then $N \wcbn W$ and $\EE \cup \{ (V,W) \} \in \X$
      ($V$ and $W$ are values);
    \item the converse of the above two conditions, on $N$; 
    \end{enumerate} 
  \item if $\EE \in \X$ then for all $(\lambda x. P, \lambda x . Q )
    \in \EE$ and for all $(M,N) \in {\starred\EE}$ it holds that
    $P\sub{M}x \Xv{\EE} Q\sub{N}x$.
  \end{enumerate}
  \emph{Environmental bisimilarity}, $\simEB$, is the largest
  environmental bisimulation.
\end{definition} 

For environmental bisimilarity to be expressed via a first-order
transition system, a
few issues have to be resolved. For instance, an environmental
bisimilarity contains both triples $(\EE, M,N)$, and pure environments
$\EE$, which shows up in the difference between clauses (1) and (2) of
Definition~\ref{d:bisim-for-lambdaN}.  Moreover, the input supplied to
tested terms may be constructed using arbitrary contexts.

We write $\LaNu$ for the first-order LTS resulting from the
translation of $\LaN$.  The states of $\LaNu$ are sequences of
$\lambda$-terms in which only the last one need not be a value.  We
use $Γ$ and $Δ$ to range over sequences of values only; thus $(Γ,M)$
indicates a sequence of $\lambda$-values followed by $M$.
We write $|Γ|$ for the length of a sequence $Γ$, and
$\Gamma_i$ for the $i$-th element in $\Gamma$, when $i\leq |Γ|$.

For a finite environment $\EE$, we write $\EE_1$ for an ordered
projection of the pairs in $\EE$ on the first component, and $\EE_2$
is the corresponding projection on the second component.  In the
translation, intuitively, a triple $(\EE, M,N)$ of an environmental
bisimulation is split into the two components $(\EE_1, M)$ and
$(\EE_2,N)$.
When $C$ is a context of arity $|Γ|$,
we write $C[Γ]$ for the term obtained by replacing each
hole $\holei i$ in $C$ with the value $\Gamma_i$.  The rules for
transitions in $\LaNu$ are as follows; they are reminiscent of~\cite{LeiferMilner00}.
\begin{equation}\label{e:cbn:lts}
  \inferrule{M\cbn M'}{(Γ,M)\ts{τ}(Γ,M')} \qquad\qquad
  \inferrule{Γ_i(C[Γ]) \cbn M'}{Γ\ts{i,C} (Γ,M')}
\end{equation}

The first rule says that if $M$ reduces to $M'$ in $\LaN$ then $M$ can
also reduce in $\LaNu$, in any environment. The second rule implements
the observations in clause (2) of
Definition~\ref{d:bisim-for-lambdaN}: in an environment $\Gamma $
(only containing values), any component $\Gamma_i$ can be tested by
supplying, as input, a term obtained by filling a context $C$ with
values from $\Gamma $ itself. The label of
the transition records the position $i $ and the context chosen.
As the rules show, the labels of $\LaNu$ include the special label
$τ$, and can also be of the form $i,C$ where $i$ is a integer and $C$
a context.

We establish full abstraction from environmental bisimilarity to
bisimilarity on $\LaNu$ for finite environments. Full abstraction for
the empty environment alone is enough for our interests since
contextual equivalence corresponds to environmental bisimilarity with
the empty environment. One could accommodate $\LaNu$ and the
corresponding full abstraction result for possibly-infinite
environments, however we felt that it was not worth the notational
complications, since infinite environments are not reachable from
finite ones in environmental bisimulations, and since we do not think
that infinite environments increase discriminative power.  In
the statement below, $\approx$ denotes standard weak bisimilarity
(Definition~\ref{d:weakprogression}) on $\LaNu$.

The following proof shows a precise correspondence between
environmental bisimulations and bisimulations in $\LaNu$. The reader
familiar with environmental bisimilarities should find the statement
illustrative and maybe applicable to other variants of environmental
bisimilarities. It is also possible to show a direct, although less
precise, correspondence between contextual equivalence and
bisimilarity. This second approach is shown for the imperative
$\lambda$-calculus in Section~\ref{s:icbv} and exploits the compatibility
of up-to-context functions. Since compatibility of up-to-context is
proved independently of the correspondence result for $\LaNu$, this
approach would also work for $\LaNu$; however,  we found it more
interesting here  to  show the more precise result.

\begin{theorem}\label{t:cbn:fullabstraction}
  When $\EE$ is a finite environment,
  \[
    ~~M \simEBv{𝓔} N~~
    ⇔
    ~~(𝓔_1,M) \approx (𝓔_2,N)~~
    {\quad}\mbox{ and }{\quad}
    ~~𝓔 ∈ {\simEBv{}}~~
    ⇔
    ~~𝓔_1 \approx 𝓔_2~~ \enspace.
  \]
\end{theorem}
\begin{proof}
  \textbf{(⇒)} We show that if $𝓧$ is an environmental bisimulation
  then $𝓧^2$ is a (first-order) weak bisimulation, where $𝓧^2$ relates
  $(𝓔_1,M)$ to $(𝓔_2,N)$ when $(𝓔,M,N)∈𝓧$, and $𝓔_1$ to $𝓔_2$ when
  $𝓔∈𝓧$. By symmetry we consider only one direction: we suppose
  $x\ 𝓧^2\ y$ and a transition $x\ts{μ}x'$, and we obtain $y'$ such
  that $y\ws{\hat{μ}}y'$ and $x'\ 𝓧^2\ y'$.
  \begin{enumerate}
  \item $μ=τ$: then $x=(𝓔_1,M)\ts{τ}(𝓔_1,M')=x'$ with $M\cbn M'$, and
    $y=(𝓔_2,N)$ with $M\ 𝓧_{𝓔}\ N$. By definition of environmental
    bisimulation, $N\wcbn N'$ with $M'\ 𝓧_{𝓔}\ N'$ and hence
    $y \ws{} y'$ with $y'=(𝓔_2,N')$ and $x'\ 𝓧^2\ y'$.
  \item $μ=i,C$: then $(x,y)=(𝓔_1,𝓔_2)$ with $𝓔∈𝓧$, and
    $x'=(𝓔_1,P\subst{C[𝓔_1]}{x})$ with $λ x.P=(𝓔_1)_i$ and we choose
    $y'=(𝓔_2,Q\subst{C[𝓔_2]}{x})$ with $λ x.Q=(𝓔_2)_i$. Then by
    construction, $y\ts{i,C}y'$ and $x'\ 𝓧^2\ y'$ because
    $(λ x.P,λ x.Q)∈𝓔$ and $(C[𝓔_1],C[𝓔_2])∈𝓔^\star$.
  \end{enumerate}
  
  \textbf{(⇐)} The correspondence is less direct, so instead of
  establishing a correspondence between weak bisimulations, we define
  the candidate relation on top of weak bisimilarity. We first
  write $Γ⋅Δ$ for the \emph{pairing} of $Γ$ and
  $Δ$, i.e.\ the relation $\{(Γ_i,Δ_i) \mid i \leq
  |Γ|,|Δ|\}$. The environmental relation $𝓧$ is defined as follows:
  \begin{align*}
    𝓧 ~\eqdef~ \{(Γ⋅Δ,M,N) ∣ (Γ,M) ≈ (Δ,N)\}
    ~∪~ \{Γ⋅Δ ∣ Γ ≈ Δ\}
  \end{align*}
  (where $Γ$ and $Δ$ only contain values). We prove that $𝓧$ is an
  environmental bisimulation.
  \begin{enumerate}
  \item Suppose $M\ 𝓧_{Γ\cdot Δ}\ N$ (i.e.\ $(Γ,M)≈(Δ,N)$).
    \begin{enumerate}
    \item if $M \cbn M'$ then $(Γ,M) \ts{τ} (Γ,M')$, which implies
      $(Δ,N) \ws{} (Δ,N')$ with $(Γ,M') ≈ (Δ,N')$ and hence
      $N \wcbn N'$ with $M'\ 𝓧_{Γ\cdot Δ}\ N'$;
    \item if $M=V$, we  need a $W$ such that $N \wcbn W$ and
      $Γ\cdot Δ ∪ \{ (V,W) \} ∈ 𝓧$. Since $(Γ,V) = x \ts{i,C} x'$ for
      some $x'$, we have $(Δ,N) \ws{i,C} y_3$, for some
      $y_3$, i.e.\ $(Δ,N) =y_0 \ws{}y_1\ts{i,C}y_2\ws{} y_3$. Since
      $y_1$ has an $i,C$ transition, $y_1$ is of the form $(Δ,W)$ for
      some $W$. Since $y_0\ws{}y_1$ and $\ts{τ}$ is deterministic, we
      derive $y_0≈y_1$. By transitivity of $≈$, we infer $x≈y_1$, hence
      $(Γ,V)≈(Δ,W)$. We can then conclude $(Γ,V)\cdot (Δ,W)∈𝓧$.
    \item the converse of the above two conditions, on $N$, holds, as
      $\approx$ is symmetric.
    \end{enumerate}
  \item If $(\lambda x. P, \lambda x . Q )∈Γ\cdot Δ \in \X$ and
    $(M,N) \in {\starred{(Γ\cdot Δ)}}$, we prove that
    $P\subst{M}x \Xv{Γ\cdot Δ} Q\subst{N}x$.
    
    We have $(\lambda x. P, \lambda x . Q) = (Γ_i,Δ_i)$
    for some $i$, and $(M,N)=(C[Γ],C[Δ])$ for some $C$. Then
    $Γ_i(C[Γ]) \cbn P\subst{M}{x}$ so $Γ \ts{i,C} (Γ,P\subst{M}{x})$
    must be answered with $Δ \ws{i,C} (Δ,N')$ and
    $(Γ,P\subst{M}{x})≈(Δ,N')$. There are no silent transitions coming
    from $Δ$ so we necessarily have the $\ts{i,C}$ transition first.
    Since there is only one such transition, we have in fact
    $Δ \ts{i,C} (Δ,Q\subst{N}{x}) \ws{} (Δ,N')$. Again, as
    ${\ws{}} \subseteq {\approx}$,  by transitivity of $\approx$
    we derive $(Γ,P\subst{M}{x})≈(Δ,N')≈(Δ,Q\subst{N}{x})$, and hence
    $(P\subst{M}{x})\ 𝓧_{Γ\cdot Δ}\ Q\subst{N}{x}$.
  \end{enumerate}
\end{proof}

The theorem also holds for the strong versions of the
bisimilarities. Again, having established full abstraction with
respect to a first-order transition system and ordinary bisimilarity,
we can inherit the theory of bisimulation enhancements. We have
however to check up-to techniques that are specific to environmental
bisimilarity.

\subsection*{Structure and reusability of proofs}
The first technique 
is proved compatible in \linebreak[4]Lemma~\ref{l:cbn:weakening}, which is an
example of the standard way of proving compatibility. The other three
techniques 
are interdependent in that they each progress to a function containing
all three (Lemmas~\ref{l:progression:term}, \ref{l:progression:env},
and \ref{l:progression:eval}). These progressions could be established
separately, which would be an improvement of modularity over a
monolithic proof of compatibility (itself an improvement of size over
two redundant proofs of up-to-context and congruence). Moreover, we
achieve here a substantial amount of additional proof refactoring
thanks to two general ingredients. The first
(Definition~\ref{d:values:nonvalues}, Lemmas~\ref{l:env:eval:values}
and~\ref{l:val:nonval}) may be of general interest to handle calculi
whose grammars separate  `values' from `non-value'. The second
(Lemmas~\ref{l:red:expansion}, \ref{l:trick:more:general},
and~\ref{l:trick}) may be of general interest for calculi that are
quasi-deterministic,  in the sense of
Definition~\ref{d:quasideterministic}. (These results are
used again in Section~\ref{s:icbv}.)  The three progressions are finally
combined into Theorem~\ref{t:cbn:context}.

\bigskip

A useful technique specific to environmental bisimilarity is
`up-to-environment', which allows us to replace an environment with a
larger one. We define $\weak (\R)$ as the smallest relation that
includes $\R$ and such that, whenever
$(V,\Gamma,M) \mathrel{\weak(\R)} (W, \Delta,N)$ holds, also
$(\Gamma,M) \mathrel{\weak(\R)} (Δ,N)$ holds, where $V$ and $W$ are
any values. Here $\weak$ stands for `weakening' as, from
Lemmas~\ref{l:compatible:closure} and~\ref{l:cbn:weakening}, if
$(V,Γ,M)≈(W,Δ,N)$ then $ (Γ,M)≈(Δ,N)$.

\begin{lemma}\label{l:cbn:weakening}
  Function $\weak$ is compatible.
\end{lemma}
\begin{proof}
  Since silent transitions do not alter the environment, we only consider
  ($i,C$)-transitions; writing $Γ'=V_1,…,V_n,Γ$ and $Δ'=W_1,…,W_n,Δ$, we
  have:
  \begin{mathpar}
    Γ_i(C[Γ]) = Γ'_{i+n}(C_{+n}[Γ])  \and
    Δ_i(C[Δ]) = Δ'_{i+n}(C_{+n}[Δ]) \enspace,
  \end{mathpar}
  where $C_{+n}$ is $C$ where each hole $\holei{j}$ has been replaced
  with $\holei{j+n}$. Then $Γ\ts{i,C}(Γ,M')$ implies
  $Γ'\xrightarrow{i+n,C_{+n}}(Γ',M')$ and $Δ'\ws{i+n,C_{+n}}(Δ',N')$
  implies $Δ\ws{i,C}(Δ,N')$, and so from $𝓡 \relprogress{} 𝓢$ we obtain
  $\weak(𝓡) \relprogress{} \weak(𝓢)$.
\end{proof}

Somewhat dual to weakening is the strengthening of the environment, in
which a component of an environment can be removed. However this is
only possible if the component removed is `redundant', that is, it can
be obtained by gluing other pieces of the environment within a
context; strengthening is captured by the following $\env$ function:
\[ \env(𝓡) ~\defi~ \big\{((Γ,C_v[Γ],M) ,(Δ,C_v[Δ],N) ) ~\st~ (Γ,M) \R
  (Δ,N) \} \]%
where $C_v$ ranges over value contexts (i.e., the outermost operator
of $C_v$ is an abstraction or $C_v$ is a hole). We show that $\env$ is
below the companion in Theorem~\ref{t:cbn:context}.

For up-to-context, we need to distinguish between arbitrary contexts
and evaluation contexts.  There are indeed congruence properties,
and corresponding up-to techniques, that only hold for the latter
contexts.  A hole $\holei i$ of a context $C$ is in a \emph{redex
  position} if the context obtained by filling all the holes but
$\holei i$ with values is an {evaluation context}.  Below, $C$ ranges
over arbitrary contexts, whereas $E$ ranges over contexts in which
the first hole $\holei 1$ appears exactly once \emph{and} in redex position.
\[\begin{array}{rcll}
  \term(𝓡) &\defi& \big\{((Γ,C[Γ])  ,(Δ,C[Δ])  ) &\st  Γ \R 
  Δ
  \,\big \}
    \\
  \eval(𝓡) &\defi& \big\{((Γ,E[M,Γ]), (Δ,E[N,Δ])) &\st
  (Γ,M) \R (Δ,N)
  \big\}
\end{array}\]
We will prove that functions $\term$, $\env$, and $\eval$ are below
both companions with a separate progression result for each function.
We start by establishing a progression for $\term$.


\begin{lemma}
  \label{l:progression:term}
  $\term \funprogress{} T(\env ∪ \term ∪ \eval)$.
\end{lemma}
\begin{proof}
  Suppose that $𝓡 \relprogress{} 𝓢$, we show that
  $\term(𝓡) \relprogress{} T(\env ∪ \term ∪ \eval)(𝓢)$. More
  explicitly, we show that
  $\term(𝓡) \relprogress{} \env(\S) ∪ \env(\term(\R)) ∪ \term(\R) ∪
  \eval(\S)$.

  Let $Γ\R Δ$ and $n=|Γ|$. We analyse transitions from $(Γ,C[Γ])$.
    \begin{enumerate}
    \item Suppose $C[Γ]$ is a value, so $C$ is a value context $C_v$;
      the transition to consider is of the form $(Γ,C_v[Γ]) \ts{i,C} (Γ,C_v[Γ],M')$. 
      \begin{enumerate}
      \item If $i \leq n$ then $Γ_i(C[Γ,C_v[Γ]]) \cbn M'$ and since
        $C[Γ,C_v[Γ]] = C'[Γ]$ with $C' = C[-,C_v]$, we obtain
        $Γ \ts{i,C'} (Γ,M')$, and similarly for $Δ$.
        \begin{center}
          \begin{tikzpicture}[baseline,descr/.style={fill=white,inner sep=2pt}]
            \matrix (m) [matrix of math nodes, row sep=2em, column sep=3em]
            { Γ & Δ \\ (Γ,M') & (Δ,N') \\ };
            \path[-](m-1-1) edge node[descr] {$ 𝓡 $} (m-1-2);
            \path[-](m-2-1) edge node[descr] {$ 𝓢 $} (m-2-2);
            \path[->](m-1-1) edge node[auto,swap] {$ i,C' $} (m-2-1);
            \draw(m-1-2) edge[-implies,double equal sign distance] node[auto]{$i,C'$} (m-2-2);
          \end{tikzpicture}
          $~~\leadsto~~$
          \begin{tikzpicture}[baseline,descr/.style={fill=white,inner sep=2pt}]
            \matrix (m) [matrix of math nodes, row sep=2em, column sep=5em]
            { (Γ,C_v[Γ]) & (Δ,C_v[Δ]) \\ (Γ,C_v[Γ],M') & (Δ,C_v[Δ],N') \\ };
            \path[-](m-1-1) edge node[descr] {$ \term(𝓡) $} (m-1-2);
            \path[-](m-2-1) edge node[descr] {$ \env(𝓢) $} (m-2-2);
            \path[->](m-1-1) edge node[auto,swap] {$ i,C $} (m-2-1);
            \draw(m-1-2) edge[-implies,double equal sign distance] node[auto]{$i,C$} (m-2-2);
          \end{tikzpicture}
        \end{center}
      \item If $i = n+1$ and $C_v = \holei{j}$, the same argument as
        above applies, replacing $Γ_i$ with $Γ_j$ and $i,C'$ with
        $j,C'$.
      \item If $i = n+1$ and $C_v$ is not a hole, then $M'$ is of the
        form $C'[Γ]$; then $Δ$ makes the same transition to
        $(Δ,C_v[Δ],C'[Δ])$ and
        $((Γ,C_v[Γ],C'[Γ]),(Δ,C_v[Δ],C'[Δ])) \in \env(\term(𝓡))$.
     \end{enumerate}
   \item If $C[Γ]$ is not a value then $C$ is necessarily of the form
     $C=E[C_{v 1} C_2, -]$, for some evaluation context $E$, value
     context $C_{v 1}$, and context $C_2$. The transition is of the
     form $(Γ,C[Γ]) \ts{τ} (Γ,M')$ with $C[Γ] \cbn M'$. We distinguish
     two cases:
      \begin{enumerate}
      \item $C_{v 1}$ is not a hole (i.e.\ $C_{v 1} = λx.C_1[x,-]$ for
        some $C_1$). Then $M'=C'[Γ]$ for $C' = E[C_1[C_2,-],-]$ and
        $C[Δ] \cbn C'[Δ]$. The resulting pair $((Γ,C'[Γ]),(Δ,C'[Δ]))$
        is in $\term(𝓡)$.
      \item $C_{v 1} = \holei{i}$ and so $C[Γ]=E[Γ_i(C_2[Γ]),Γ]$. Then
        $M' = E[M_1,Γ]$ for some $M_1$ such that
        $Γ_i(C_2[Γ]) \cbn M_1$. Using the label $i,C_2$, the
        progression $𝓡 \relprogress{} 𝓢$ provides us with an answer
        $(Δ,N_1)$ such that $Δ_i(C_2[Δ]) \wcbn N_1$, which allows us
        to conclude up to $\eval$:
      \end{enumerate}
    \end{enumerate}
    \begin{center}
      \hfill
      \begin{tikzpicture}[baseline,descr/.style={fill=white,inner sep=2pt}]
        \matrix (m) [matrix of math nodes, row sep=2em, column sep=3em]
        { Γ & Δ \\ (Γ,M_1) & (Δ,N_1) \\ };
        \path[-](m-1-1) edge node[descr] {$ 𝓡 $} (m-1-2);
        \path[-](m-2-1) edge node[descr] {$ 𝓢 $} (m-2-2);
        \path[->](m-1-1) edge node[auto,swap] {$ i,C_2 $} (m-2-1);
        \draw(m-1-2) edge[-implies,double equal sign distance] node[auto]{$i,C_2$} (m-2-2);
      \end{tikzpicture}
      $~~\leadsto$
      \begin{tikzpicture}[baseline,descr/.style={fill=white,inner sep=2pt},remember picture]
        \matrix (m) [matrix of math nodes, row sep=2em, column sep=5em]
        { (Γ,E[Γ_i(C_2[Γ]),Γ]) & (Δ,E[Δ_i(C_2[Δ]),Δ]) \\ (Γ,E[M_1,Γ]) & |[alias=baselinenode]| (Δ,E[N_1,Δ]) \\ };
        \path[-](m-1-1) edge node[descr] {$ \term(𝓡) $} (m-1-2);
        \path[-](m-2-1) edge node[descr] {$ \eval(𝓢) $} (m-2-2);
        \path[->](m-1-1) edge node[auto,swap] {$ τ $} (m-2-1);
        \draw(m-1-2) edge[-implies,double equal sign distance] node[auto]{$τ$} (m-2-2);
      \end{tikzpicture}
      \hspace*{-12mm} 
      \hfill
      \begin{tikzpicture}[remember picture]
        \coordinate (right paper edge);
        \node[overlay,anchor=base east,inner sep=0pt]
        at (baselinenode.base -| right paper edge)
        {\qedhere};
      \end{tikzpicture}%
    \end{center}
\end{proof}

Before moving on to the techniques $\env$ and $\eval$, it is useful to
remark that when they are applied to values, they look like special
cases of $\term$. This can be used to shorten the proofs
substantially, but this needs to be made formal first by defining a
restriction function and using it to relate $\env$ and $\eval$ to
$\term$.

\newcommand{\valuespairs}{\mathsf{v}}
\newcommand{\nonvaluespairs}{\mathsf{n}}
\begin{definition}
  \label{d:values:nonvalues}
  Let $\mathcal{V}$ be the set of value configurations (of form
  $\Gamma$) and $\overline{\mathcal{V}}$ the set of non-value
  configurations, i.e.~sequences for which the last term is not a
  value (of form $(\Gamma,M)$ where $M$ is not a value). We define now
  two restriction functions on relations:
  \begin{align*}
    \valuespairs(\RR) &\eqdef \mathcal R \cap (\mathcal{V}\times \mathcal{V})
    \\\nonvaluespairs(\RR) &\eqdef \mathcal R \cap (\overline{\mathcal{V}}\times \overline{\mathcal{V}})
  \end{align*}
\end{definition}


The first step is to show that indeed, techniques $\eval$ and $\env$
are, on value configuration pairs, special cases of $\term$:

\begin{lemma}
  \label{l:env:eval:values}
  $\eval ∘ \valuespairs \subseteq t ∘ \term$ and
  $\env ∘ \valuespairs \subseteq t ∘ \term$.
\end{lemma}
\begin{proof}
  Any pair in $\eval(\valuespairs(\RR))$ is of the form
  $((\Gamma', E[\Gamma_{n},\Gamma']),(\Delta', E[\Delta_n,\Delta']))$
  where: $n$ is the arity of $E$, $(\Gamma,\Delta)\in {\R}$, and
  $\Gamma'$ (respectively $\Delta'$) is the sequence $\Gamma$
  (respectively $\Delta$) without its last element. The context
  $C=E[\holei{n},\holei{1},\dots,\holei{n-1}]$ applied to
  $(\Gamma,\Delta)\in {\R}$ shows that the original pair is of the
  form $((\Gamma', C[\Gamma]),(\Delta', C[\Delta]))$ and hence is in
  $t(\term(\RR))$: we use $\weak \subseteq t$ to remove the $n$th
  values from the environments $Γ$ and $Δ$. The same argument applies
  for $\env$ as well, except that we use $t$ in $t(\term(\R))$ only to
  swap the last two elements the sequences.
\end{proof}

We handled pairs of value configurations, so now we need to handle the
other kinds of pairs. We first handle the case where the left member
of the pair is a value configuration. 
We need however to first define a notion of determinism of an LTS:

\begin{definition}
  \label{d:quasideterministic}
  We say that a LTS $(\pr, \Act, \longrightarrow )$ is
  \emph{quasi-deterministic} if there exists an equivalence relation
  $≃$ on $\Act$ such that for all labels $μ,μ'\in\Act$ and processes
  $x,x_1,x_2\in \pr$ (where $x\ts{μ}$ is short for $(∃x'~x\ts{μ}x')$):
  \begin{enumerate}
  \item \label{i:det:tau} $μ ≃ τ$ implies $μ = τ$,
  \item \label{i:det:sim} $x\ts{μ} x_1$ and $x\ts{μ}x_2$ imply $x_1 \sim x_2$,
  \item \label{i:det:lab} $x\ts{μ}$ and $x\ts{μ'}$ imply $μ ≃ μ'$,
  \item \label{i:det:rec} $x\ts{μ}$ and $μ ≃ μ'$ implies $x\ts{μ'}$.
  \end{enumerate}
\end{definition}
This version of determinism is looser than strict determinism, since
it allows derivatives to be strongly bisimilar and not necessarily
equal, and labels to be related through some equivalence relation,
rather than equal. This equivalence relation must in turn be reflected
by the set of labels that can be performed from a given process.

A similar notion can be found in the formalisation of a compiler with
some non-determinism~\cite{compcertTSO}, where such a relation on
labels is defined. This relation satisfies~\eqref{i:det:tau}, a LTS that
is said to be `determinate' satisfies~\eqref{i:det:sim}
(although~\eqref{i:det:sim} is more relaxed as it allows for bisimilar
processes) and~\eqref{i:det:lab} and a `receptive' LTS
satisfies~\eqref{i:det:rec}.

\begin{lemma}
  \label{l:red:expansion}
  In a quasi-deterministic LTS, ${\ts{τ}} \subseteq {≳}$.
\end{lemma}
\begin{proof}
  We show that $(\ts{τ}) ⊆ \expb(≳)$. Let $x,y$ such
  that $x\ts{τ}y$.
  \begin{itemize}
  \item If $x\ts{μ}x'$, then
     $μ≃τ$      by~\eqref{i:det:lab},
     $μ=τ$      by~\eqref{i:det:tau},
     $x'\sim y$ by~\eqref{i:det:sim},
     so in particular $x' ≳ y$.
     Hence, the challenge can be answered with
     $y \ws{\hat{μ}} y$.
   \item If $y\ts{μ}y'$, then $x\ws{μ}y'$, so we conclude by
     reflexivity of $≳$.
  \end{itemize}
  We conclude by remarking that $\expb(≳) ⊆ (≳)$.
\end{proof}

\begin{lemma}
  \label{l:trick:more:general}
  In a quasi-deterministic LTS,
  \newcommand{\notau}{\not\!\!\!{\ts{τ}}}
  if $(x,y) ∈ \weakb(𝓢)$ and $x\ts{μ}$ with $μ≠τ$, then for some
  $y_1$, $y \ws{} y_1 \ts{\mu}$ with $(x,y_1) ∈ \weakb (≳\S≲)$.
\end{lemma}
\begin{proof}
  We first prove that whenever $μ_1,μ_2 ≠ τ$, for all $x_2,x_2'$,
  \begin{align}
    \label{e:sim:red:stop}
    (x_1 \ws{} x_1'\ts{μ_1}{}) ∧
    (x_2 \ws{} x_2'\ts{μ_2}{}) ∧
    x_1 \sim x_2 ⇒
    x_1' \sim x_2'
  \end{align}
  by induction on $x_1\ws{}x_1'$.
  \begin{itemize}
  \item If $x_1 = x_1'$, it is enough to show that $x_2 = x_2'$.
    Suppose otherwise that $x_2 \ws{} x_2'$ takes at least one step,
    and so $x_2 \ts{τ}{}$. Since $x_1 \sim x_2$, $x_1 \ts{τ}{}$, and so
    by \eqref{i:det:lab}, $μ_1 ≃ τ$, and by \eqref{i:det:tau},
    $μ_1 = τ$ (contradiction).
  \item Suppose now $x_1 \ts{τ} x_1' \ws{} x_1''$, and that the
    induction hypothesis holds for $x_1'$. Since $x_1 \sim x_2$, we
    can derive $x_2'$ such that $x_2 \ts{τ} x_2'$ and
    $x_1' \sim x_2'$. The transition $x_2 \ws{} x_2''$ must take at
    least one step to some $x'$, otherwise by \eqref{i:det:lab},
    $μ_2 ≃ τ$ and then by \eqref{i:det:tau}, $μ_2 = τ$
    (contradiction). By \eqref{i:det:sim}, $x' \sim x_2'$, and by
    transitivity and symmetry of bisimilarity, $x_1' \sim x'$, so we
    conclude by induction.
  \end{itemize}
  We have established~\eqref{e:sim:red:stop}.
  
  Since $x\ts{μ}$,
  $(x,y)∈\weakb(S)$ provides us with $y_1,y_1',y_1''$ such that
  $y \ws{} y_1 \ts{\hat{μ}} y_1' \ws{} y_1''$. Because $μ≠τ$,
  $y_1 \ts{μ} y_1'$, so there only remains to prove that
  $(x,y_1) \in \weakb(≳\S≲)$. More precisely we will prove that
  $(x,y_1) \in \weakb(\S\sim\Leftarrow)$, which entails
  $(x,y_1) \in \weakb(≳\S≲)$ by Lemma~\ref{l:red:expansion}.
  Note that by \eqref{e:sim:red:stop}, 
  for all $α≠τ$ and $y_0$,
  $y \ws{} y_0 \ts{α}$ implies $y_1 \sim y_0$ $(*)$.
  We now
  prove~\eqref{e:LtoR}, which we will use twice.
  \begin{align}
    \label{e:LtoR}
    x \ts{α} x_2 ~~⇒~~ ∃ y_2'' ~y_1' ~y_1'' ~~~ y_1\ts{α}y_1'\ws{}y_1''\sim y_2'' ~∧~ x_2 \S y_2''
  \end{align}
  By \eqref{i:det:lab} and~\eqref{i:det:tau}, $α≠τ$. $\weakb(\S)$
  provides us again with $y_2,y_2',y_2''$ such that
  $y \ws{} y_2 \ts{α} y_2' \ws{} y_2''$ with $x_2 \S y_2''$. By $(*)$,
  $y_1 \sim y_2$, from which we can play the transitions
  $y_2\ts{α}\ws{} y_2''$ to obtain $y_1''$ such that
  $y_1 \ts{α}\ws{} y_1''$ with $y_1''\sim y_2''$, which ends the proof
  of~\eqref{e:LtoR}.
  
  We finally show $(x,y_1) \in \weakb(\S\sim\Leftarrow)$:
  \begin{itemize}
  \item Suppose that $x \ts{α} x_2$. Using \eqref{e:LtoR}, we can
    answer the challenge with $y_1''$. We indeed have
    $y_1 \ws{α} y_1''$ and $x_2 \S\sim y_1''$, hence
    $x_2 \S\sim\Leftarrow y_1''$.
  \item Suppose now that $y_1 \ts{α} y_3'$. By \eqref{i:det:lab},
    $μ ≃ α$, and by \eqref{i:det:rec}, $x \ts{α} x_2$ for some $x_2$.
    We use $x_2\ts{α}x_2'$ as our weak transition. We can now use
    \eqref{e:LtoR} again. By \eqref{i:det:sim}, $y_1'\sim y_3'$. Since
    $y_1'\ws{} y_1''$, there is $y_3''$ such that $y_3'\ws{} y_3''$
    and $y_1''\sim y_3''$. We conclude by transitivity of $\sim$ since
    $y_2' \sim y_1'' \sim y_3''$.
    \qedhere
  \end{itemize}
\end{proof}

\begin{lemma}
  \label{l:trick}
  In $\LaNu$, if $(Γ,y) ∈ \weakb(𝓢)$ then $y \ws{} Δ$ with
  $(Γ,Δ) ∈ \weakb (≳\S≲)$.
\end{lemma}
\begin{proof}
  $\LaNu$ is quasi-deterministic, using $i,C≃i',C'$ whenever $C$ and
  $C'$ are of the same arity, so we use
  Lemma~\ref{l:trick:more:general} with $x=Δ$, $μ=1,C_0$ (where $C_0$
  is a context of arity $|Γ|$, for example the context $λx.x$ with no
  hole), which gives us $y\ws{}y'\ts{μ}{}$, hence $y'$ is of the form
  $Δ$, with $(x,y')∈ \weakb (≳\S≲)$.
\end{proof}

\begin{remark}
  \label{r:aboutthetrick}
  Lemma~\ref{l:trick} does not apply to non-deterministic calculi. In
  fact, those would require a special label signalling that the
  configuration is only composed of values in order for the proofs of
  progressions to go through. This would make Lemma~\ref{l:trick}
  unnecessary in the proofs of progressions for $\env$ and $\eval$
  (those proofs would however need to include long parts that are redundant
  with the proof of progression for $\term$
  since Lemma~\ref{l:val:nonval} uses
  Lemma~\ref{l:trick}).
\end{remark}

\newcommand{\reductions}{\mathsf{r}}%
Lemma~\ref{l:val:nonval} helps separating a proof of progression
$f \leadsto T(f∪g)$ for a function $f$ into a few simpler proofs, namely:
\eqref{e:fv:compat}, that is,  the progression for pairs of values;
\eqref{e:fn:compat},  that is, the progression for pairs of non-values;
and \eqref{e:fr}, that is,  the fact that $f$ absorbs the reduction
function~$\reductions\eqdef (𝓡↦{\Rightarrow\R\Leftarrow})$, up to $t$.
In order to carry out the splitting, $f$ is also required to
distribute over union~\eqref{e:fdistr}~--- which holds for
$\env$, $\term$, and $\eval$.


\begin{lemma}
  \label{l:val:nonval}
  If $f$ and $g$ are monotone functions such that:
  \begin{align}
    \label{e:fv:compat} f ∘ \valuespairs & \leadsto T(f∪g)\\
    \label{e:fn:compat} f ∘ \nonvaluespairs & \leadsto T(f∪g)\\
    \label{e:fr} f ∘ \reductions & \subseteq t ∘ f\\
    \label{e:fdistr} f(𝓡∪𝓢) & \subseteq f(𝓡)∪f(𝓢)
  \end{align}
  then $f \leadsto T(f∪g)$.
\end{lemma}
\begin{proof}
  We first establish the following inclusion:
  \begin{align}
    \label{e:red}
    (\mathsf{id}\setminus \nonvaluespairs) ∘ \bb \subseteq \reductions ∘ \valuespairs ∘ \bb ∘ t
  \end{align}
  Let $\S$ be a relation and
  $(x,y)\in (\id \setminus \nonvaluespairs)(\bb(\S))$,
  i.e.~$(x,y)\in \bb(\S)$ and at least one of $x$ or $y$ is a value.
  The case $\bb=\strongb$ is trivial: since $x$ is a value if and only if
  $y$ is a value, they are both values, and
  $(x,y)\in \valuespairs(\bb(\S)) \subseteq
  \reductions(\valuespairs(\bb(t(\S))))$. The case $\bb=\weakb$ is a
  consequence of Lemma~\ref{l:trick}:
  \begin{itemize}
  \item First suppose that $x$ is a value $\Gamma$. By
    Lemma~\ref{l:trick}, there is a value $\Delta$ such that
    $(Γ,Δ) ∈ \weakb (≳\S≲)\subseteq \weakb(t(\S))$. This is a value
    pair, so $(Γ,Δ) ∈ \valuespairs(\weakb(t(\S)))$. Finally,
    $(x,y) \in \reductions(\valuespairs(\weakb(t(\S))))$.
  \item Otherwise, suppose $y$ is a value $\Delta$. We know that
    $(\Delta,y) \in \weakb(\S)^{-1} = \weakb(\S^{-1})$ by symmetry of
    $\weakb$, so we can apply Lemma~\ref{l:trick}. This shows that
    there exists $\Gamma$ such that $(Δ,Γ) ∈ \weakb (≳\S^{-1}≲)$.
    Following the reasoning  for $x$, we derive
    $(y,x) \in \reductions(\valuespairs(\weakb(t(\S^{-1}))))$, and so
    $(x,y) \in \reductions(\valuespairs(\weakb(t(\S))))$ since
    $\reductions$, $\valuespairs$, $\weakb$, and $t$ are symmetric.
  \end{itemize}
  We can now conclude:
  \[\begin{array}{r@{\,}ll}
    f ∘ \bb &= f ∘ (\nonvaluespairs \cup (\id \setminus \nonvaluespairs)) ∘ \bb \\
    &\subseteq f ∘ \nonvaluespairs ∘ \bb \cup f∘ (\id \setminus \nonvaluespairs) ∘ \bb &\mbox{ by~\eqref{e:fdistr}} \\
    &\subseteq \bb ∘ T(f∪g) \cup f∘ (\id \setminus \nonvaluespairs) ∘ \bb &\mbox{ by \eqref{e:fn:compat}}\\
    &\subseteq \bb ∘ T(f∪g) \cup f∘ \reductions ∘ \valuespairs ∘ \bb ∘ t &\mbox{ by \eqref{e:red} and monotonicity of $f$}\\
    &\subseteq \bb ∘ T(f∪g) \cup t∘ f ∘ \valuespairs ∘ \bb ∘ t &\mbox{ by \eqref{e:fr}}\\
    &\subseteq \bb ∘ T(f∪g) \cup t∘ \bb ∘ T(f∪g) ∘ t &\mbox{ by \eqref{e:fv:compat} and monotonicity of $t$}\\
    &\subseteq \bb ∘ T(f∪g) \cup \bb∘ t ∘ T(f∪g) ∘ t &\mbox{ by compatibility of $t$}\\
    &= \bb ∘ T(f∪g) &\mbox{ since $t \subseteq T(h)$ and $T(h)^3 = T(h)$ for all $h$}
    \end{array}
    \tag*{\qedhere}
  \]
\end{proof}

The distinction between values and non-values simplifies  the proof of the
progression for $\env$.

\begin{lemma}
  \label{l:progression:env}
  $\env \funprogress{} T(\env ∪ \term ∪ \eval)$.
\end{lemma}
\begin{proof}
  This is the conclusion of Lemma~\ref{l:val:nonval} with $f=\env$ and
  $g=\term ∪ \eval$, so it is sufficient to establish the premises of
  the lemma:
  \begin{enumerate}
  \item $\env ∘ \valuespairs \leadsto T(\env ∪ \term ∪ \eval)$:
    \[\begin{array}{r@{\,}ll}
        \env ∘ \valuespairs ∘ \bb
        &⊆ t ∘ \term ∘ \bb & \mbox{by Lemma~\ref{l:env:eval:values}} \\
        &⊆ t ∘ \bb ∘ T(\env ∪ \term ∪ \eval) & \mbox{by Lemma~\ref{l:progression:term} and monotonicity of $t$} \\
        &⊆ \bb ∘ t ∘ T(\env ∪ \term ∪ \eval) & \mbox{by compatibility of $t$} \\
        &⊆ \bb ∘ T(\env ∪ \term ∪ \eval) & \mbox{since $t⊆ T(h)$ and $T(h)^2=T(h)$ for all $h$.} \\
      \end{array}\]
  \item $\env ∘ \nonvaluespairs \leadsto T(\env ∪ \term ∪ \eval) $
    follows from the stronger inclusion
    $\env ∘ \nonvaluespairs \leadsto \env $:
    
    Let ${\R}\relprogress{}{\S}$ and
    $((Γ,C_v[Γ],M),(Δ,C_v[Δ],N)) \in \str(\nonvaluespairs(\R))$ i.e.~with
    $M$ and $N$ non-values. Challenges from the left-hand side are of
    the form $(Γ,C_v[Γ],M) \ts{\tau} (Γ,C_v[Γ],M')$, which is
    equivalent to $(Γ,M) \ts{\tau} (Γ,M')$. The progression
    ${\R}\relprogress{}{\S}$ tells us that there exists $N'$ such that
    $(Δ,N) \ws{} (Δ,N')$ with $(Γ,M) \S (Δ,N')$, and so
    $(Δ,C_v[Δ],N) \ws{} (Δ,C_v[Δ],N')$ with
    $((Γ,C_v[Γ],M'),(Δ,C_v[Δ],N')) \in \str(\S)$. Challenge from the
    right-hand side are handled symmetrically.
  \item $\env ∘ \reductions ⊆ t ∘ \env$: since
    $\reductions ⊆ \uptoexp ⊆ t$, we only need to prove
    $\env ∘ \reductions ⊆ \reductions ∘ \env $. This can be derived
    more algebraically: it is trivial to check that $\env$ respects
    relation mirroring, relational composition, and silent transitions
    ($\env(\R^{-1}) ⊆ \env(\R)^{-1}$, $\env(\R\S) ⊆ \env(\R)\env(\S)$,
    and $\env(\ts{τ})⊆{\ts{τ}}$), from which
    $\env(⇒\R⇐) \subseteq {⇒}\env(\R){⇐}$ is a direct consequence.
    \qedhere
  \end{enumerate}
\end{proof}

The stronger result $\env \funprogress{} T(\env ∪ \term)$
holds as well, but it requires a  longer proof (also  redundant with
the progression for $\term$) and it is not necessary.
The progression for $\eval$ follows the same pattern.

\begin{lemma}
  \label{l:progression:eval}
  $\eval \funprogress{} T(\env ∪ \term ∪ \eval)$.
\end{lemma}
\begin{proof}
  Similarly we apply Lemma~\ref{l:val:nonval} with $f=\eval$ and
  $g=\term ∪ \env$, and prove the hypotheses:
  \begin{enumerate}
  \item $\eval ∘ \valuespairs \leadsto T(\env ∪ \term ∪ \eval)$ is a
    consequence of Lemmas~\ref{l:env:eval:values}
    and~\ref{l:progression:term}.
  \item $\eval ∘ \nonvaluespairs \leadsto T(\env ∪ \term ∪ \eval) $
    follows from the stronger inclusion
    $\eval ∘ \nonvaluespairs \leadsto \eval $. This can be proved the
    same way as in Lemma~\ref{l:progression:env} using the facts that
    $E[M,Γ] \cbn M_1$ implies that for some $M'$, $M_1 = E[M',Γ]$ with
    $M\cbn M'$ and that $N\wcbn N'$ implies $E[N,Δ] \wcbn E[N',Δ]$.
  \item $\eval ∘ \reductions ⊆ t ∘ \eval$ is similarly derived from
    $\eval(\R^{-1}) ⊆ \eval(\R)^{-1}$,
    $\eval(\R\S) ⊆ \eval(\R)\eval(\S)$, and $\eval(\ts{τ})⊆{\ts{τ}}$.
    \qedhere
  \end{enumerate}
\end{proof}

\begin{theorem}
  \label{t:cbn:context}
  The functions $\env, \term, \eval$ are below both companions
  $t_{\strongb}$ and $t_{\weakb}$.
\end{theorem}
\begin{proof}
  Combining Lemmas~\ref{l:progression:term}, \ref{l:progression:env},
  and \ref{l:progression:eval} provides us with the following
  progression for $\weakb$:
  \begin{align*}
    \env ∪ \term ∪ \eval &\funprogress{} T(\env ∪ \term ∪ \eval)
  \end{align*}
  and therefore $\env ∪ \term ∪ \eval$ is below the companion
  $t_\weakb$. The case for $\strongb$ is similar but easier; in
  particular the analogue of Lemma~\ref{l:trick} is not required.
\end{proof}

Once more, the fact that up-to-context functions are below $t$
entails the corresponding congruence properties of environmental
bisimilarity. In \cite{envbisim} the two aspects (congruence and
up-to-context) had to be proved separately, with similar proofs.
Moreover the two cases of contexts (arbitrary contexts and evaluation
contexts) had to be considered at the same time, within the same
proof. Here, in contrast, the machinery of compatible functions allows
us to split the effort into simpler proofs.

\begin{remark}\label{r:cancellation}
  A transition system ensuring full abstraction as in
  Theorem~\ref{t:cbn:fullabstraction} does not guarantee the
  compatibility of the up-to techniques specific to the language in
  consideration. For instance, a simpler and maybe more natural
  alternative to the second transition in \eqref{e:cbn:lts} is the
  following one:
  \begin{equation}
    \label{e:cbn:lts:alt}
    \inferrule{ }{Γ\ts{i,C} (Γ, Γ_i(C[Γ]))}
  \end{equation}
  With this rule, full abstraction holds, but up-to-context is
  unsound: for every $Γ$ and $Δ$, the singleton relation $\{(Γ,Δ)\}$ is
  a bisimulation up to $\term$: indeed, using
  rule~\eqref{e:cbn:lts:alt}, the derivatives of the pair $Γ,Δ$ are of
  the shape $Γ_i(C[Γ])$, $Δ_i(C[Δ])$, and they can be discarded
  immediately, up to the context $\holei i C$. If up-to-context were
  sound then we would deduce that any two terms are bisimilar. (The
  rule in~\eqref{e:cbn:lts} prevents such a behaviour since it ensures
  that the tested values are `consumed' immediately.)
\end{remark}


\section{Imperative call-by-value λ-calculus}
\label{s:icbv}

In this section we study the addition of imperative features
(higher-order references, that we call locations), to a call-by-value
$\lambda$-calculus. It is known that finding powerful reasoning
techniques for imperative higher-order languages is a hard problem.
The language, $\LaI$, is a simplified variant of that in
\cite{KoutavasW06,envbisim}.  The syntax of terms, values, and
evaluation contexts, as well as the reduction semantics are given in
Figure~\ref{f:LR}.  A $\lambda$-term $M$ is run in a \emph{store}: a
partial function from locations to closed values, whose domain
includes all free locations of both $M$ and its own co-domain.
We use letters $r,s,u,v$ to range over stores. New store locations may be
created using the operator $\newloc M$; the content of a store
location $\ell$ may be read using $\get{ℓ} V$, or rewritten using
$\set{ℓ} V$ (the argument of the former instruction is ignored, and
the latter instruction returns the identity value
$I \eqdef \lambda x. x$). We denote the reflexive and transitive
closure of $\icbv$ by $\wicbv$.

Note that in contrast with the languages in
\cite{KoutavasW06,envbisim}, locations are not directly first-class
values; the expressive power is however the same: a first-class
location ${ℓ}$ can always be encoded as the pair $(\get{ℓ},\set{ℓ})$.
Having locations as first-class values by themselves is possible but
would require two additional labels (for reading and writing), two
additional rules, two new cases in the corresponding case analyses,
and new ways to build contexts from environments; presentation and
proofs would then be substantially more involved. Hence, for
readability issues, we have preferred to forbid it.

\begin{figure}[t] 
  \begin{align*}
    M &::= x ∣ M M ∣ \newloc M ∣ V&&&
    V &::= λx.M ∣  \get{ℓ} ∣ \set{ℓ}&&&
    E &::= \holei{} ∣ E V ∣ M E
  \end{align*}
  \begin{mathpar}
    \inferrule{ }{(s\sep (λx.M)V) \icbv (s\sep M\{V/x\}) } \and
    \inferrule{ ℓ∉\dom{s} }{(s\sep \newloc M) \icbv (s[ℓ↦I]\sep M) } \and
    \inferrule{ ℓ∈\dom{s} }{(s\sep \get{ℓ}V) \icbv (s\sep s[ℓ]) } \and
    \inferrule{ ℓ∈\dom{s} }{(s\sep \set{ℓ}V) \icbv (s[ℓ↦V]\sep I) } \and
    \inferrule{ (s\sep M) \icbv (s'\sep M') }{(s\sep E[M]) \icbv (s'\sep E[M']) } \and
  \end{mathpar}
  \caption{The imperative $\lambda$-calculus} 
  \label{f:LR}
\end{figure}

We present the first-order LTS for $\LaI$, and then we relate the
resulting strong and weak bisimilarities directly with contextual
equivalence (the reference equivalence in $\lambda$-calculi).
Alternatively, we could have related the first-order bisimilarities to
the environmental bisimilarities of $\LaI$, and then inferred the
correspondence with contextual equivalence from known results about
environmental bisimilarity, as we did for $\LaN$.

We write $(s\sep M) \dwa$ when $M$ is a value; and $(s\sep M) \Dwa$ if
$(s\sep M) \wicbv \dwa$.  For the definition of contextual
equivalence, we distinguish the cases of values and of arbitrary
terms, because they have different congruence properties: values
can be tested in arbitrary contexts, while arbitrary terms must be
tested only in evaluation contexts.  As in~\cite{envbisim}, we
consider contexts that do not contain free locations (they can contain
bound locations). We refer to \cite{envbisim} for more details on
these aspects.

\begin{definition}\label{d:refeq:values}
  \begin{itemize}
  \item For values $V$, $W$, we write $(s\sep V) \refeq (r\sep W)$
    when $(s\sep C[V]){⇓}$ iff $(r\sep C[W]){⇓}$, for all
    location-free contexts $C$.
  \item For terms $M$ and $N$, we write $(s\sep M) \refeq (r\sep N)$
    when $(s\sep E[M]){⇓}$ iff $(r\sep E[N]){⇓}$, for all
    location-free evaluation contexts $E$.
  \end{itemize}
\end{definition}

We now define $\LaIu$, the first-order LTS for $\LaI$.  The states and
the transitions for $\LaIu$ are similar to those for the pure
$\lambda$-calculus of Section~\ref{s:cbn}, with the addition of a
component for the store.  The two transitions (\ref{e:cbn:lts}) of
call-by-name $\lambda$-calculus become:
\begin{align*}
  \inferrule{(s\sep M)\icbv (s'\sep M')}{(s\sep Γ,M)\ts{τ}(s'\sep Γ,M')} \quad
  \inferrule{Γ'=Γ,\getset{r} \quad
    \left(s⊎r[Γ']\sep  Γ_i(C[Γ'])\right) \icbv (s'\sep M')}
     { (s\sep Γ) \arrR{i,C,\cod(r)} (s'\sep Γ',M') }
\end{align*}
The first rule is the analogous of the first rule in
(\ref{e:cbn:lts}). The important differences are on the second rule.
First, since we are \emph{call-by-value}, $C$ now ranges over
$\vcontexts$, the set of \emph{value contexts} (i.e., holes or
contexts of the form $λx.C'$) without free locations. Moreover, since
we are now \emph{imperative}, in a transition we must permit the
creation of new locations, and a term supplied by the environment
should be allowed to use them. In the rule, the new store is
represented by $r$ (whose domain has to be disjoint from that of $s$).
Correspondingly, to allow manipulation of these locations from the
observer, for each new location $\ell$ we make $\get{\ell}$ and
$\set{\ell}$ available, as an extension of the environment; in the
rule, these are collectively written $\getset{r}$, and $\Gamma'$ is
the extended environment. Finally, we must initialise the new store,
using terms that are created out of the extended environment
$\Gamma'$; that is, each new location $\ell$ is initialised with a
term $D_\ell [\Gamma']$ (for $D_\ell \in \vcontexts$). Moreover, the
contexts $D_\ell$ chosen must be made visible in the label of the
transition. To take care of these aspects, we view $r$ as a
\emph{store context}, a tuple of assignments $ℓ↦ D_\ell$. Thus the
initialisation of the new locations is written $r[\Gamma']$; and,
denoting by $\cod(r)$ the tuple of the contexts $D_\ell$ in $r$, we
add $\cod(r)$ to the label of the transition. Note also that, although
$C$ and $D_ℓ$ are location-free, their holes may be instantiated with
terms involving the $\get{\ell}$ and $\set{\ell}$ operators,
so that these contexts may still manipulate the store.

Once more, on the (strong and weak) bisimilarities that are derived
from this first-order LTS, we can import the theory of compatible
functions and bisimulation enhancements.  Like in Section~\ref{s:pi}
for $\pi$, we establish the validity of a few up-to techniques before
proving full abstraction: these techniques give us important closure
properties of bisimilarities via Lemma~\ref{l:compatible:closure}.

Concerning additional up-to functions, specific to $\LaIu$, the
functions $\weak$, $\env$, $\term$ and $\eval$ are adapted from
Section~\ref{s:cbn} in the expected manner---contexts $C_v$, $C$ and
$E$ must be location-free.  A further function for $\LaIu$ is
$\alloc$, which manipulates the store by removing locations that do
not appear elsewhere (akin to garbage collection); thus, $\alloc(𝓡)$
is the set of all pairs
\[ ((s ⊎ r[Γ']\sep Γ', M) ,\, (u ⊎ r[Δ']\sep Δ', N)) \] such that
$ (s\sep Γ, M) \R (u\sep Δ, N)$, and with $Γ'=Γ,\getset{r}$ and
$ \Delta'= \Delta ,\getset{r}$. Note that we must have
$\dom{r}∩\dom{s}=∅=\dom{r}∩\dom{u}$. This may seem unnecessarily
restrictive, but since renaming locations on either side using an injective
substitution is a strongly bisimilar operation, using
$(𝓡↦{\sim\R\sim}) ∘ \alloc$ allows to choose $r_1$ on the left and
$r_2$ on the right, as long as $\cod(r_1)=\cod(r_2)$.

\begin{lemma}\label{l:icbv:context}
  The functions $\weak, \env, \eval, \alloc, \term$ are below both
  companions  $t_\strongb$ and  $t_\weakb$.
\end{lemma}
\begin{proof}
  We apply the same proof schema as in
  Theorem~\ref{t:cbn:context} with more technical details to be
  handled, as the store is to be accounted for. We provide
  details mainly for the progression starting from $\alloc$ itself,
  which is the most interesting new aspect. We
  explain how the other parts
  are handled, with reference to  the proof of
  Theorem~\ref{t:cbn:context}. \medskip
  
  We handle $\alloc$ first. To avoid introducing and remembering many
  new names such as $Γ'$, $Γ''$, etc., we write $Γ^V$ for $Γ,V$ and
  $Γ^r$ for $Γ,\getset{r}$. For example, the rule for visible
  transitions can be rewritten into
  \[\inferrule
    {(s\sep Γ_i(C[Γ^r])) \icbv (s'\sep M')}
    {(s,Γ)\xrightarrow{i,C,\cod(r)}(s'\uplus r[Γ^r],Γ^r,M')}
    \enspace.
  \]
  It also simplifies writing
  and reading when taking one index of a composed environment, for
  example $\Gamma^{r V}_i$ should be read as the $i$th element of
  $(\Gamma^r)^V$, which can be either $\Gamma_i$ (if $i\leq|Γ|$), or
  in $\getset{r}$, or $V$.
  Now $\alloc$ can be redefined as
  \[\inferrule
    {(s\sep Γ,M) \R (u\sep Δ,N)}
    {(s\uplus r[Γ^r]\sep Γ^r,M)~\alloc(\R)~(u\uplus r[Δ^r]\sep Δ^r,N)}
    \enspace.
  \]
  
  We assume ${\R}\relprogress{}{\S}$ and analyse the transitions
  starting from pairs in $\alloc(\R)$, i.e.\ the transitions of
  $(s\uplus r[Γ^r]\sep Γ^r,M)$. Silent transitions are, once again,
  easy to handle, as the locations of $M$ are contained in the domain
  of $s$, and the other part of the term, namely $r[Γ^r]$, is left
  unchanged (progressing to $\alloc(𝓢)$---in particular,
  $\alloc ∘ \nonvaluespairs \funprogress{} \alloc$).

  We now handle the visible transitions of
  $(s\uplus r[Γ^r]\sep Γ^r,M)$, (i.e.\ $M$ is a value $V$), labelled
  by $μ$ such that $μ=i,C,\cod(v)$ for some $v$. We choose $v$ such
  that the locations used by $v$, which do not appear in the label,
  are fresh. The transition is:
  \begin{align}
    \label{e:icbvstore:tr}
    \inferrule
    { (s\uplus r[Γ^r] \uplus v[Γ^{r V v}]\sep Γ^{r V}_i(C[Γ^{r V v}])) \icbv (s' \sep M')}
    {(s\uplus r[Γ^r]\sep Γ^{r V}) \xrightarrow{i,C,\cod(v)} (s'\sep Γ^{r V v}, M')}
  \end{align}
  There are two cases, depending if $Γ^{r V}_i$ is in $Γ^V$ or in $\getset{r}$.
  \begin{enumerate}
  \item Suppose $i \leq |Γ|$ or $i=|Γ^{r}|+1=|Γ|+2|r|+1$. Then
    $Γ^{r V}_i = Γ^V_{i'}$ for $i' = \min(i,|Γ^r|+1)$, and we can
    derive a similar $\icbv$ transition from $(s,Γ^V)$, using label
    $μ' = i',C',\cod({v'})$ for some ${v'}$ and $C'$ such that:
    \begin{enumerate}
    \item $r[Γ^r] \uplus v[Γ^{r V v}] = {v'}[Γ^{V {v'}}]$
    \item $C[Γ^{r V v}] = C'[Γ^{V {v'}}]$
    \end{enumerate}
    The premise of~\eqref{e:icbvstore:tr} is hence equal to the
    premise below, which has however a different conclusion:
    \[
    \inferrule
    { (s\uplus {v'}[Γ^{V {v'}}]\sep Γ^{V}_{i'}(C'[Γ^{V {v'}}])) \icbv (s' \sep M')}
    {(s \sep Γ^{V}) \xrightarrow{i',C',\cod({v'})} (s'\sep Γ^{V {v'}}, M')} \enspace.
    \]
    We can derive the corresponding 
    transition labelled $i',C',\cod({v'})$ from $(u \sep Δ, N)$ which
    will silently reduce to $(u' \sep Δ^W)$, then make a visible weak
    transition to $(u''\sep Δ^{V {v'}}, N')$ knowing that
    $(s'\sep Γ^{V {v'}}, M') \S (u''\sep Δ^{W {v'}}, N')$. We can then
    replace $Γ^{V {v'}}$ and $Δ^{W {v'}}$ with $Γ^{r V v}$ and $Δ^{r W v}$,
    to prove that
    $(s'\sep Γ^{r V v}, M') \S_1 (u''\sep Δ^{r W v}, N')$, where
    ${\S_1}$ is $\S$ where we applied the `up-to-permutation'
    technique to move $V$ and $W$ in the middle of ${v'}$. This technique
    is compatible, so ${\S_1} \subseteq t_{\weakb}(\S)$ (and
    ${\S_1} \subseteq t_{\strongb}(\S)$).
  \item Suppose $i ∈ \{|Γ|+1,…,|Γ|+2|r|\}$. Then $Γ^{r V}_i$ is either
    $\lget{ℓ}$ or $\lset{ℓ}$ with $ℓ∈\dom{r}$.
    \begin{enumerate}
    \item if $Γ^{r V}_i=\lget{ℓ}$ then $s'$ is not modified and
      $M'=r[Γ^r]_ℓ$ is a context of $Γ^r$ and hence is also a context of $Γ^{V {v'}}$ using
      the same ${v'}$ as above. The result of the transition is:
      $$(s \uplus {v'}[Γ^{V {v'}}]\sep Γ^{V {v'}}, C_1[Γ^{V {v'}}]) \eqdef x'.$$
      In the weak case, using Lemma~\ref{l:trick} we get
      $(u\sep Δ, N) \ws{} (u'\sep Δ^W)$, and this term is  related to $(s,Γ^V)$ through
      $\weakb(≳\S≲)⊆t_{\weakb}(𝓢)$.
      
      
      Finally we can relate $x'$ to
      $(u' \uplus {v'}[Δ^{W {v'}}]\sep Δ^{W {v'}}, C_1[Δ^{W {v'}}])$ through
      $\term(\alloc(t(𝓢)))$.
    \item if $Γ^{r V}_i=\lset{ℓ}$ then $s'$ is modified at $ℓ∈\dom{r}$
      (so we only have to change  $r$) and $M'=I=C_1[Γ^{V {v'}}]$ for
      $C_1 = I$ (a context with no holes)
      so the pair progresses again, using the same notations as
      before, to $\term(\alloc(t(𝓢)))$.
    \end{enumerate}
  \end{enumerate}
  In summary, we have
  $\alloc(𝓡) \relprogress{} \alloc(𝓢) ∪ t(𝓢) ∪ \term(\alloc(t(𝓢)))$, and so
  \begin{align}
    \label{e:icbv:alloc}
    \alloc \funprogress{} T(\alloc ∪ \term)
  \end{align}
  We now establish the progressions for the remaining functions.
  First, $\weak\funprogress{}\weak$ with the same argument as in the
  proof of Lemma~\ref{l:cbn:weakening}, so $\weak \subseteq t$. The
  most important proof is for $\term$. We assume
  $(s\sep Γ)\R(u\sep Δ)$, and we analyse the transitions from
  $(s\sep Γ,C[Γ])$.
  \begin{enumerate}
  \item if $C[Γ]$ is a value, then $C$ is a value context $C_v$, and
    the same structure as for the corresponding case in
    Lemma~\ref{l:progression:term} applies here, with the only
    significant difference being in the third case. The transition of
    interest is labelled with $i,C_1,\cod(r)$ such that $i\leq|Γ|+1$.
    \begin{enumerate}
    \item If $i\leq |Γ|$, this means the
      value that is given an argument is one of the
      $Γ_i$s. Let $i,C_1',\cod(r')$ be the label $i,C_1,\cod(r)$ where we
      composed the contexts with $C_v$, so that $C_v$ replaces $\holei{|Γ|+1}$.
      Using this label on the progression $𝓡\relprogress{}𝓢$ we obtain
      a pair in $𝓢$. We apply first $\env$ and then `up-to-permutation' to add
      $C_v$ on each side, which puts the desired
      pair in $t(\env(𝓢))$.
    \item If $i=|Γ|+1$, and $C_v=\holei{j}$, we proceed the same way
      as above, with the label $j,C_1',\cod(r')$ where $C_1'$ (resp. $r'$)
      is $C_1$ (resp. $r$) where $\holei{j}$ replaces all occurrences of
      $\holei{|Γ|+1}$.
    \item If $i=|Γ|+1$ and $C_v$ is not a hole ($C_v=λx.C_2[x,-]$), then the derivative
      is $(s\uplus r[Γ^{\prime r}]\sep Γ^{\prime r},C_3[Γ])$ with $Γ' = Γ,C_v[Γ]$
      and $C_3=C_2[C_v,-]$
      and with an augmented store,
      which results
      in a relation built on $𝓡$, as follows 
      (we use $\term$ and  $\alloc$, and set
      $Δ' = Δ,C_v[Δ]$):
      \[
      \inferrule{
        \inferrule{
          \inferrule
          {(s\sep Γ)~𝓡~(u\sep Δ)}
          {(s\sep Γ')~\term(𝓡)~(u\sep Δ')}}
        {(s\uplus r[Γ^{\prime r}]\sep Γ^{\prime r})~\alloc(\term(𝓡))~(u \uplus r[Δ^{\prime r}]\sep Δ^{\prime r})}}
      {(s\uplus r[Γ^{\prime r}]\sep Γ^{\prime r},C_2[Γ])~\term(\alloc(\term(𝓡)))~(u \uplus r[Δ^{\prime r}]\sep Δ^{\prime r},C_2[Γ])}
      \enspace.
      \]
    \end{enumerate}
  \item If $C[Γ]$ is not a value, then either
    $C=E[C_{v 1} C_{v 2}, -]$ or 
    $C=E[\newloc C_1[-, \get{ℓ}, \set{ℓ}],-]$.
    \begin{enumerate}
    \item If $C=E[C_{v 1} C_{v 2}, -]$ and $C_{v 1}$ is not a hole,
      then $C_{v 1} = λx.C_1[x,-]$ for some $C_1$. The transition is
      of the form $(s;Γ,C[Γ])\ts{τ} (s;Γ,C'[Γ])$ with
      $C'=C_1[C_{v 2},-]$, and similarly for $Δ$:
      $(u;Δ,C[Δ])\ts{τ} (u;Δ,C'[Δ])$. This pair of derivatives is in
      $\term(𝓡)$.
    \item If if $C=E[C_{v 1} C_{v 2}, -]$ and $C_{v 1}=\holei{i}$,
      then some $Γ_i$ is run, so we also run it starting from the
      original configuration, with the label $i,C_{v 2},∅$ using the
      evaluation context function, and therefore progressing to
      $\eval(𝓢)$.
    \item The most interesting case is when
      $C=E[\newloc C_1[-, \get{ℓ}, \set{ℓ}],-]$. Then, $C[Γ]$ creates
      a private location, i.e.,
      $$C[Γ]=E[\newloc C_1[Γ, \get{ℓ}, \set{ℓ}],Γ]$$ and
      $(s\sep Γ,C[Γ])\ts{τ} (s⊎[ℓ↦I] \sep Γ, C_2[Γ, \get{ℓ},
      \set{ℓ},Γ])$ with $C_2=E[C_1,-]$. We prove a stronger result, namely that
     the resulting configurations with the
      context $C_2$ 
are still related if  the $\lget{ℓ}$
      and $\lset{ℓ}$ operators  are available.
The derivation is as follows; 
we use       weakening $\weak$,  exploit $\term$ and $\alloc$, and  
 write
      $Λ$ for $\get{ℓ},\set{ℓ}$:
      \[
      \inferrule{
        \inferrule{
          \inferrule
          {(s\sep Γ) ~~ 𝓡 ~~ (u\sep Δ)}
          {(s⊎[ℓ↦I]\sep Λ,Γ) ~~ \alloc(𝓡) ~~ (u⊎[ℓ↦I]\sep Λ,Δ)}}
        {(s⊎[ℓ↦I]\sep Γ,Λ,C_2[Γ,Λ]) ~~ \term(\alloc(𝓡)) ~~ (u⊎[ℓ↦I]\sep Δ,Λ,C_2[Δ,Λ])}}
      {(s⊎[ℓ↦I]\sep Γ,C_2[Γ,Λ]) ~~ \weak(\term(\alloc(𝓡))) ~~ (u⊎[ℓ↦I]\sep Δ,C_2[Δ,Λ])}\enspace.
      \]
    \end{enumerate}
  \end{enumerate}
  To summarise,
  \[ \term(𝓡) \relprogress{} t(\env(𝓢)) ∪ \term(\alloc(\term(𝓡))) ∪
    \term(𝓡) ∪ \eval(𝓢) ∪ \weak(\term(\alloc(𝓡))) \] and the
  right-hand side is included in
  $T(\weak ∪ \env ∪ \term ∪ \alloc ∪ \eval)(𝓡 ∪ 𝓢)$. Remark that $\weak ⊆ t$, $𝓡 ⊆ \bb(𝓢)$,
  and $\bb⊆t$, so we obtain:
  \begin{align}
    \label{e:icbv:term}
    \term \funprogress{} T(\env ∪ \term ∪ \alloc ∪ \eval)
    \enspace.
  \end{align}
  We now move on to $\env$ and $\eval$. As in $\LaNu$,
  $\valuespairs(𝓡)$ denotes the pairs of value configurations of $𝓡$,
  and $\nonvaluespairs(𝓡)$ the pairs of non-value configurations of
  $𝓡$. It is trivial to check that
  $\env ∘ \valuespairs \subseteq t ∘ \term$ and
  $\eval ∘ \valuespairs \subseteq t ∘ \term$, and so by combining
  with~\eqref{e:icbv:term}, both $\env∘ \valuespairs$ and
  $\eval∘ \valuespairs$ progress to
  $T(\env ∪ \term ∪ \alloc ∪ \eval)$.
  
  It is also straightforward to derive
  $\env ∘ \reductions \subseteq \reductions ∘ \env$ and
  $\eval ∘ \reductions \subseteq \reductions ∘ \eval$, and that
  $\env ∘ \nonvaluespairs \funprogress{} \env$ and
  $\eval ∘ \nonvaluespairs \funprogress{} \eval$. Note that $\LaIu$ is
  quasi-deterministic; indeed, new locations, both for $\newloc$ and
  in the choice of the domain of $r$ in visible transitions, are
  chosen non-deterministically, but their choice does not matter up to
  strong bisimilarity. We can now apply Lemma~\ref{l:val:nonval} with
  $f=\env$ and $g=\term ∪ \alloc ∪ \eval$ to obtain:
  \begin{align}
    \label{e:icbv:env}
    \env \funprogress{} T(\env ∪ \term ∪ \alloc ∪ \eval)
  \end{align}
  and with $f=\eval$ and
  $g=\env ∪ \term ∪ \alloc$ to obtain:
  \begin{align}
    \label{e:icbv:eval}
    \eval \funprogress{} T(\env ∪ \term ∪ \alloc ∪ \eval)
  \end{align}
  Combining ~\eqref{e:icbv:alloc}, %
  ~\eqref{e:icbv:term}, %
  ~\eqref{e:icbv:env}, and %
  ~\eqref{e:icbv:eval}, %
  yields that $h≜(\env ∪ \term ∪ \alloc ∪ \eval)$ progresses to $T(h)$,
  and hence $h\subseteq t$.
\end{proof}

Having established that $\term$ and $\eval$ are below the companion
gives as a consequence that our first-order bisimilarity is a
congruence under location-free contexts, from which we can derive the
soundness implication in Theorem~\ref{t:icbv:fullabstraction} below.



\begin{theorem}\label{t:icbv:fullabstraction}
  $(s\sep M)≡(u\sep N)$ iff $(s\sep M) ≈ (u\sep N)$.
\end{theorem}
\begin{proof}
  \textbf{(⇐)} The function $\eval$ is below $t$ by
  Lemma~\ref{l:icbv:context}. By Lemma~\ref{l:compatible:closure} we
  know that $\eval(≈) \subseteq {≈}$. In other words $≈$ is a
  $\eval$-congruence, and in particular $(s\sep M) ≈ (u\sep N)$
  implies $(s\sep E[M]) ≈ (u\sep E[N])$ for every location-free
  evaluation context $E$. This in turn implies that $(s\sep E[M])$ and
  $(u\sep E[N])$ have the same weak visible transitions, which in turn
  implies that $(s\sep E[M]){⇓}$ iff $(u\sep E[N]){⇓}$.
  
  \medskip

  \textbf{$(⇒)$ } For completeness, we prove
  that the following relation $\R$  is a weak bisimulation, where $E$ ranges over
  location-free evaluation contexts:
  \begin{align}
    \label{e:completeness}
    𝓡 \eqdef \big\{ ((s\sep Γ,M),(u\sep Δ,N)) ~\st~ ∀E~
      (s\sep E[M,Γ]){⇓} \mbox{ iff } (u\sep E[N,Δ]){⇓}
    \big\}
    \enspace.
  \end{align}
  Suppose $(s\sep Γ,M)\R (u\sep Δ,N)$.
  Since $𝓡$ is symmetric, we only look at the transitions labelled with $\mu$
 emanating from   $(s\sep Γ,M)$.
  
  When $μ=τ$, it holds that $(s\sep M)\icbv(s'\sep M')$.  We then have
  $(s\sep E[M,Γ]){⇓}$ iff $(s'\sep E[M',Γ]){⇓}$ by (quasi-) determinism of
  $\icbv$, so we can conclude $(s'\sep Γ,M') \R (u\sep Δ,N)$ to close
  the bisimulation diagram.
  
  We now suppose that $μ≠τ$, i.e.\ $M$ is a value $V$, $μ=i,C,\cod(r)$, and
  $(s\sep Γ,M) \ts{μ}(s'\sep Γ'',M')$ for some $s'$, $Γ''$, $M'$
  satisfying
    \[\def\arraystretch{1.3}
    \begin{array}{cccc}
      Γ'=Γ,V ~~~& Γ''=Γ',\getset{r} ~~~& (r[Γ'']⊎s \sep Γ'_i(C[Γ'']))\icbv(s' \sep M')~~~~&(*)\
    \end{array}\]
  Since $M$ is a value,
  $(s\sep M)⇓$. By choosing $E=\holei 1$ in
  \eqref{e:completeness} we know that $(u\sep N){⇓}$ and thus
  $(u\sep N)\wicbv(u'\sep W)$ for some value $W$ and store $u'$. We  then obtain the
  weak transition $(u\sep Δ,N)\ws{μ}(u''\sep Δ'',N')$ through
  $(u'\sep Δ,W)$, for some $Δ'$, $Δ''$, $N'$ such that:
    \[\def\arraystretch{1.3}
    \begin{array}{cccc}
      Δ'=Δ,W ~~~& Δ''=Δ',\getset{r} ~~~& (r[Δ'']⊎u'\sep Δ'_i(C[Δ'']))\wicbv(u''\sep N')~~~~&(**)
    \end{array}\]
  To close the bisimulation diagram, we will now prove that
  $(s'\sep Γ'',M') \R (u''\sep Δ'',N')$. Let $E$ be a location-free
  evaluation context, we show that
    \begin{align}
      \label{e:completeness:final}
      (s' \sep E[M',Γ'']) ⇓ ~~\mbox{ iff }~~
      (u''\sep E[N',Δ'']) ⇓ \enspace.
    \end{align}
    Observe that if $(s_1, M_1) \icbv (s_1', M_1')$ then
    $(s_1, E[M_1,Γ_1]) \icbv (s_1', E[M_1',Γ_1])$, which implies
    $(s_1, E[M_1,Γ_1]){⇓} ⇔ (s_1', E[M_1',Γ_1]){⇓}$ by determinism of
    $\icbv$. Using this observation and each reductions in $(*)$ and
    $(**)$, \eqref{e:completeness:final} becomes equivalent to:
    \begin{align}
      \label{e:completeness:final2}
      (r[Γ''] ⊎ s \sep E[Γ'_i(C[Γ'']),Γ'']) ⇓ ~~\mbox{ iff }~~
      (r[Δ''] ⊎ u'\sep E[Δ'_i(C[Δ'']),Δ'']) ⇓
    \end{align}
    We recall that a context of \emph{arity} $n$ is a context with
    holes $\holei{1},\dots,\holei{n}$ each occurring any number of
    times and that in an evaluation context the first hole
    $\holei{1}$ is the one that occurs exactly once and in evaluation
    position. Let $F$ be an evaluation context of arity $|Γ|+1$.
    Instantiating the definition of $𝓡$ with $F$, we have the
    following equivalence:
    \begin{align}
      \label{e:completeness:final3}
      (s \sep F[M,Γ]) ⇓ ~~\mbox{ iff }~~
      (u \sep F[N,Δ]) ⇓
    \end{align}
    We choose $F$ carefully so that \eqref{e:completeness:final3} is
    equivalent to~\eqref{e:completeness:final2}. Let $ℓ_i↦C_i$,
    $i= 1,\dots,n$, be the collection of location-context pairs of the
    store context $r$. Let $C'≜E[\holei i(C),-]$ i.e. $E$ where the
    evaluation hole is replaced with the context $\holei i(C)$. The
    contexts $C_1,\dots, C_n$, $C$, and $C'$, are all of arity
    $|Γ''| = |Γ|+1+2n$.
    \newcommand{\bah}{\ensuremath{\bullet}}%
    For every context $D$, let $D^\bah$ be $D$ with the following
    replacements:
    \begin{enumerate}
    \item the holes $\holei{|Γ|+1}$ are replaced with $x$,
    \item the holes $\holei{|Γ|+2 i}$ are replaced with $\get{ℓ_i}$,
    \item the holes $\holei{|Γ|+2 i+1}$ are replaced with $\set{ℓ_i}$,
    \item all holes $\holei{i}$, $i\geq 1$ are simultaneously replaced
      with $\holei{i+1}$. This shift leaves $\holei{1}$ unused.
    \end{enumerate}
    Then $C_1^\bah,\dots, C_n^\bah$, $C^\bah$, ${C'}^\bah$ are of arity
    $|Γ|+1$, with no occurrence of $\holei{1}$. 
    \newcommand*{\cmd}[1]{\ensuremath{~\mathtt{#1}~}}%
    We now define the evaluation context $F$ of arity $|Γ|+1$ as follows:
    \[F \eqdef \cmd{let} x = \holei{1} \cmd{in}
    νℓ_1…νℓ_n~ℓ_1:=C_1^\bah;~…;ℓ_n:=C_n^\bah;~{C'}^\bah
    \]
    
    After $1+2n$ steps of reductions (one for the substitution of $x$
    with $M=V$, one for each $νℓ_i$, one for each assignment)
    $(s;F[M,Γ])$ reduces to $(r[Γ''] ⊎ s \sep E[Γ'_i(C[Γ'']),Γ''])$.
    Similarly $(u;F[N,Δ])$ first reduces to $(u';F[W,Δ])$, and then
    reduces to $(r[Δ''] ⊎ u' \sep E[Δ'_i(C[Δ'']),Δ''])$. By
    determinism of reductions, since each side
    of~\eqref{e:completeness:final3} reduces to the corresponding side
    of~\eqref{e:completeness:final2}, we know that
    \eqref{e:completeness:final3}
    is equivalent to~\eqref{e:completeness:final2}.
\end{proof}

Congruence of bisimilarity is restricted either to
values ($\term$), or to evaluation contexts ($\eval$). It does not
hold for arbitrary contexts, but Lemma~\ref{l:safe:any} provides a
sufficient condition for some relations between arbitrary terms to be
preserved by arbitrary contexts. First we establish weaker results:
for evaluation contexts (Lemma~\ref{l:safe:closedbyeval}), then for
\emph{non}-evaluation contexts (Lemma~\ref{l:safe:noteval}). Finally
Lemma~\ref{l:safe:any} combines the two.

In the following, we use $\asymp$ to denote any of the relations
$\sim, \approx$, and $≳$. (In Lemma~\ref{l:safe:closedbyeval} $F$ may
contain free locations, unlike occurrences in earlier definitions of
transitions and of up-to-context functions.)

\begin{lemma}
  \label{l:safe:closedbyeval}
  Suppose that for all $s$ and $Γ$, we have
  $(s\sep Γ,L) \asymp (s\sep Γ,R)$. Then for all $s$, $Γ$ and
  evaluation contexts $F$ that may contain free locations,
we have   $(s\sep Γ,F[L]) \asymp (s\sep Γ,F[R])$.
\end{lemma}
\begin{proof}
  Let $A$ be the list of $\lset{ℓ}$ and $\lget{ℓ}$ for all locations
  $ℓ$ in $F$. Then   we can obtain some 
location-free $F'$
from $F$ such that
  $F=F'[-,A]$ . By hypothesis we know
  $(s\sep Γ,A,L) \asymp (s\sep Γ,A,R)$ on which we apply
  congruence for evaluation contexts $\eval$ to derive
  $$(s\sep Γ,A,F'[L,A]) \asymp (s\sep Γ,A,F'[R,A])$$
  (Lemmas~\ref{l:icbv:context} and~\ref{l:compatible:closure}). By
  weakening 
we finally obtain
  $(s\sep Γ,F'[L,A]) \asymp (s\sep Γ,F'[R,A])$.
\end{proof}

\begin{lemma}\label{l:safe:noteval}
  Let $L$, $R$ be $\LaI$ terms with $(s\sep Γ,L) \asymp (s\sep Γ,R)$
  for all environments $Γ$ and stores $s$. Suppose $C$ is a multi hole
  context with no hole  in evaluation position. Then
  for all $Γ$ and $s$ we have
  $(s\sep Γ, C[L]) \asymp (s\sep Γ, C[R])$.
\end{lemma}
\begin{proof}
  We do the proof for the most interesting case, ${\asymp} = {≳}$, and
  we discuss the other cases at the end of the proof.
  
  Let $𝓡$ relate each configuration
  $((ℓ↦C_v^ℓ[L])_ℓ\sep \tilde{C_v}[L], C[L])$ to the one where $R$
  replaces $L$, namely $((ℓ↦C_v^ℓ[R])_ℓ\sep \tilde{C_v}[R],C[R])$,  for all
  $(C_v^ℓ)_ℓ$ and $\tilde{C_v}$ families of value contexts (i.e., of the
  form $λx.C'$), and where $C$ ranges over contexts with no hole
   in evaluation position. For simplicity we write $s_L$, $s_R$,
  $Γ_L$, and $Γ_R$ for the corresponding stores and environments. 
  The transitions from both sides,
  $(s_L\sep Γ_L,C[L])$ and $(s_R\sep Γ_R,C[R])$, have
  the same shape.  We thus show
  that $𝓡$ is an expansion up to expansion. We also rely on
  the fact that $L$ and $R$   are never  run.
  
  \begin{enumerate}
  \item \textbf{} (Case of silent action.)  Since $L$ in $C[L]$ and $R$ in $C[R]$
   are not in evaluation position, both sides  perform the same kind of
    transition. The resulting configurations
    are $(s'_L\sep Γ_L,C_1[L])$ and $(s'_R\sep Γ_R,C_1[R])$ for some $C_1$. (Even
    if a $\lset{ℓ}$ or a $\lget{ℓ}$ is involved, and  some terms containing $L$ or $R$ are
    moved to or from the store, the configurations maintain the same shape.)
  
    The only part of the invariant of the relation $𝓡$ that is not
    preserved is that $L$ or $R$ may appear in evaluation position, if
    $C_1[L] = E_1[L,L]$ (where $\holei 1$ is in evaluation position and
    $\holei 2$ may appear everywhere).
    In this case, we remark that $F_1\eqdef E_1[-,L]$ is an evaluation
    context, on which we can apply Lemma~\ref{l:safe:closedbyeval} to
    yield $(s'_L\sep Γ_L,E_1[L,L]) ≳ (s'_L\sep Γ_L,E_1[R,L])$. Let
    $C_2 ≜ E_1[R,-]$. If $C_2$ is a context with no hole in evaluation
    position, we have $(s'_L\sep Γ_L,E[R,L]) \R (s'_R\sep Γ_R,E[R,R])$
    and we have closed the diagram. If not, then let $E_2$ be such
    that $C_2 = E_2[-,L]$. Applying
    Lemma~\ref{l:safe:closedbyeval} as many times as necessary, we can replace occurrences of $L$ with
    $R$, one at at time, as long as there remain holes in evaluation
    position. The progression to $≳\R$ still holds, since $≳$ is
    transitive.
  
  \item \textbf{} (Case of visible action.) First, since no hole is in
    evaluation position, $C[L]$ is a value iff $C[R]$
    is a value, so they have the same visible actions of the form
    $i,D,\cod(r)$. We end up with the same shape of configurations we
    had for the $τ$ transition above, and we therefore proceed similarly.
  \end{enumerate}  
  We have thus proved that $𝓡$ progresses to $≳\R$ (expansion up to
  expansion). In the strong case, we prove that $𝓡$ progresses to
  $\sim\R$, and in the weak case we prove that $𝓡$ weakly progresses
  to $(≈\R) \cap (\R≈)$ (which corresponds to two possible ways of
  using Lemma~\ref{l:safe:closedbyeval} in the above proof). Such a refinement is
  necessary because in the weak case, one can use ``up to $≈$'' only
  when $≈$ is not on the same side as the challenge. 
\end{proof}

\begin{lemma}
  \label{l:safe:any}
  Let $\asymp$ be any of the relations $\sim, \approx$, and $≳$.
  Suppose $L$, $R$ are $\LaI$ terms with
  $(s\sep Γ,L) \asymp (s\sep Γ,R)$ for all environments $Γ$ and store
  $s$. Then also $(s\sep Γ,C[L]) \asymp (s\sep Γ,C[R])$, for every store
  $s$, environment $Γ$ and context $C$.
\end{lemma}
\begin{proof}
  Using Lemma~\ref{l:safe:closedbyeval} and transitivity of $\asymp$,
  we rewrite the occurrence of $L$ that is in evaluation position into $R$, and repeat
  this
  until there is no such $L$ (such a rewriting may have to be performed more than once
  if $L$ is not a value but $R$ is so; for
  example if $L=I I$ and $R = I$, then $L$ is in evaluation position
  in $L L$ on the right and in $L R$ on the left). We finally apply
  Lemma~\ref{l:safe:noteval}.
\end{proof}

The separation between evaluation contexts and non-evaluation contexts
is critical, as handling all contexts together would yield a much
larger bisimulation candidate.



\begin{lemma}\label{l:safe:eval}
  Suppose that $E$ and $E'$ are evaluation contexts and that for all
  values $V$ and stores $s$, we have $(s\sep E[V])\icbv^+ (s\sep E'[V])$.
  Then for all environments $Γ$ and stores $s$, we have $(s\sep Γ,E[M]) ≳
  (s\sep Γ,E'[M])$.
\end{lemma}
\begin{proof}
  For a given $Γ$ we consider $𝓡 = \{(s\sep Γ,E[M]), (s\sep Γ,E'[M]) ∣ \mbox{for all $s$ and } M\}$
  and the transitions from both sides:
  \begin{enumerate}
  \item when $M$ is not a value, $(s\sep M)\icbv (s'\sep M')$ and the only
    transition from both sides is a silent transition, and the derivatives are still
    in the
    relation.
  \item  when $M=V$ and the challenge transition is from the term on the left-hand side, 
    by hypothesis we have  $(s\sep Γ,E[V])\,{\ts{τ}}{}^+
    (s\sep Γ,E'[V])$, so the first transition from the left-hand side is a
    $τ$.  We use up-to-expansion to reach $(s\sep Γ,E'[V])$,  which is equal
    to the right-hand side, and conclude up to reflexivity.
  \item  Suppose now $M=V$ and the challenge transition is from the term on the right-hand side.
  Then the right-hand side makes some
    transition $(s\sep Γ,E'[V])\ts{α} (s'\sep Γ',N')$. We know that
    $(s\sep Γ,E[V])\,{\ts{τ}}{}^+ (s\sep Γ,E'[V])$ so $(s\sep Γ,E'[V])\ws{}\ts{α} (s'\sep Γ',N')$
    and we conclude again up to reflexivity.
  \end{enumerate}
  We have thus  proved that $𝓡$ is an expansion relation up to expansion and
  reflexivity.
\end{proof}

\begin{corollary}\label{cor:redexp}
  Suppose that $E$ and $E'$ are evaluation contexts and that for every
  value $V$ and store $s$, we have $(s\sep E[V])\wicbv (s\sep E'[V])$.
  Then for every store $s$ and context $C$, we have $(s\sep C[E[M]]) ≳
  (s\sep C[E'[M]])$.
\end{corollary}
\begin{proof}
  This is  a consequence of Lemma~\ref{l:safe:eval} and Lemma~\ref{l:safe:any}.
\end{proof}

We use Lemma~\ref{l:safe:any} at various places in the example we cover in
Section~\ref{ss:exa}. For instance we use it to replace a term $ N_1
\defi (λx.E[x])M$ (with $E$ an evaluation context) with $ N_2 \defi
E[M]$, under an arbitrary context. Such a property is delicate to
prove, even for closed terms, because the evaluation of $M$ could
involve reading from a location of the store that itself could contain
occurrences of $N_1$ and $N_2$.


\section{An example}
\label{ss:exa}
             
We conclude by discussing an example from \cite{KoutavasW06}.  It
consists in proving a law between terms of $\LaI$ extended with
integers, operators for integer addition and subtraction, and a
conditional---those constructs are straightforward to accommodate in
the presented framework.  For readability, we also use the standard
notation for store assignment, dereferencing and sequence: $(ℓ:=M)
\eqdef \set{ℓ}M$, ${!}ℓ \eqdef \get{ℓ}I$, and $M;N \eqdef (λx.N)M$
where $x$ does not appear in $N$.  The two terms are the following
ones:
\begin{itemize}
\item $M\defi λg.\newloc \,ℓ:=0; g(\inc); \mathtt{if}~ {!}ℓ
  ~\mathtt{mod}~ 2 = 0 ~\mathtt{then}~ I ~\mathtt{else}~ Ω$
\item $N\defi λg. g(F);I$,
\end{itemize}
where $\inc \defi λz.ℓ:={!}ℓ+2$, and $F \defi λz.I$.  Intuitively,
those two terms are weakly bisimilar because the location bound by
$\ell$ in the first term will always contain an even number.


We consider two proofs of the example.  In comparison with the proof
in \cite{envbisim}: (i) we handle the original example from
\cite{KoutavasW06}, and (ii) the availability of a broader set of
up-to techniques and the possibility of freely combining them allows
us to work with smaller relations. In the first proof we work up to
the store (through the function $\alloc$) and up to expansion---two
techniques that are not available in \cite{envbisim}. In the second
proof we exploit the up-to-transitivity technique of
Section~\ref{s:back}, which is only sound for strong bisimilarity, to
further reduce the size of the relation we work with.

\paragraph{\bf First proof.}

We first employ Lemma~\ref{l:safe:any} to reach a variant similar to
that of \cite{envbisim}: we make a `thunk' out of the test in $M$, and
we make $N$ look similar. More precisely, let $\tesp \eqdef
λz. \mathtt{if}~ {!}ℓ ~\mathtt{mod}~ 2 = 0 ~\mathtt{then}~ I
~\mathtt{else}~ Ω$, we first prove that
\begin{itemize}
\item $M≈M'\eqdef λg.\newloc \,ℓ:=0; g(\inc); \tesp I$, and
\item $N≈N'\eqdef λg. g(F); F I$.
\end{itemize}
It then suffices to prove that $M'\approx N'$, which we do using the
following relation:
\[
𝓡 \defi \left\{
  \left(s\sep M', (\inc,\tesp)_{ℓ∈\tilde{ℓ}}\right),
  \left(∅\sep N', (F,F)_{ℓ∈\tilde{ℓ}}\right)
    ~\st~ ∀ℓ∈\tilde{ℓ},~~ s(ℓ) \mbox{ is even}
  \right\} \enspace.
\]
The initial pair of terms is generalised by adding any number of
private locations. Indeed~$M'$ creates a new location when applied,
and its argument can have occurrences of $M'$ that create locations of
their own. Relation $ 𝓡$ is a weak bisimulation up to $\alloc$,
$\term$ and expansion.
  We write $(s\sep Γ_{\tl})$ for the left-hand side of a pair in $𝓡$ and
  $(∅\sep Δ_{\tl})$ for the right-hand side.
  
  Consider a transition $1,C,\cod(r)$ from $M'$ and $N'$.  We write below
  $Γ'$ for $Γ_{\tl},\getset{r}$ and $Δ'$ for $Δ_{\tl},\getset{r}$.
  \begin{itemize}
  \item $(s\sep Γ_{\tl}) ~\arrR{1,C,\cod(r)}~ (s⊎r[Γ']\sep Γ',\newloc \,ℓ:=0; C[Γ'](\inc);\tesp I)$
  \item $(∅\sep Δ_{\tl}) ~\arrR{1,C,\cod(r)}~ (r[Δ']\sep Δ',C[Δ'](F);F I)$
  \end{itemize}
  In the first line, we make the configuration run two $τ$
  transitions, so that $νℓ$ and $ℓ:=0$ get executed.  Now we have a
  new store $s' = s⊎(ℓ↦0)$ (as $s'(ℓ)$ is even,  we remain within
  the bisimulation candidate).

Now the main term is $C[Γ'](\inc);\tesp I)$, which can
  be rewritten to $D[Γ_{\tl⋅ℓ},$ $\getset{r}]$ for some context $D$.  On
  the right-hand side $C[Δ'](F);F I$ can be rewritten
to   $D[Δ_{\tl⋅ℓ},\getset{r}]$ as well.  By construction $(s'\sep Γ_{\tl⋅ℓ})
  \R (∅\sep Δ_{\tl⋅ℓ})$ hence  
$$(s'⊎r[Γ']\sep Γ_{\tl⋅ℓ},  \getset{r}) \; \;
\mathrel{\alloc(𝓡)}
\; \;
  (r[Δ']\sep Δ_{\tl⋅ℓ},\getset{r}) \; .$$  Now we first apply $\term$ with context
  $D$, then weakening $\weak$ to remove $\inc$ and $\tesp$
  which do not appear outside of $D$,
  and we thus obtain
 the required pair.

 Having handled $M'$ and $N'$, we look at a transition $i,C,\cod(r)$
 coming from some $\inc$ (and $F$ on the other side). On the
 right-hand side, $(\emptyset ⊎ r[Δ']\sep F(C[Δ']))$ immediately
 $\icbv$-reduces to $(\emptyset ⊎ r[Δ']\sep I)$, discarding its
 argument. On the left-hand side, a few more steps are necessary to
 reach $I$: $(s ⊎ r[Γ']\sep \inc(C[Γ']))$ also discards its argument
 to become first $(s ⊎ r[Γ']\sep \set{ℓ}(\get{ℓ}+2))$ and then, after
 three steps, $(s' ⊎ r[Γ']\sep I)$, with $s'≜s[ℓ\mapsto s(ℓ)+2]$.
 Since $(s';Γ_{\tl}) \R (\emptyset;Δ_{\tl})$, it is enough to apply up
 to $\alloc$ (adding $r[Γ']$), $\term$ (adding $I$), and expansion
 (accounting for the extra $\tau$ steps of the left-hand side), to
 conclude. Adapting the terms accordingly, the reasoning for $\tesp$
 is the same: $(s ⊎ r[Γ']\sep \tesp(C[Γ']))$ reduces to
 $(s ⊎ r[Γ']\sep I)$ since $s(ℓ)$ is even.

  %
  %

\paragraph{\bf Second proof.}

We first preprocess the terms using Lemma~\ref{l:safe:any}, to add
a few artificial internal steps to $N$, so that we can carry out the
reminder of the proof using strong bisimilarity, which enjoys more
up-to techniques than weak bisimilarity:
\begin{itemize}
\item $M≈M'\eqdef λg.\newloc \,ℓ:=0; g(\inc); \tesp I$,
\item $N≈N''\eqdef λg.I;I;g(\incz); \tesz I$.
\end{itemize}
where $\incz$ and $\tesz$ are pure functions that return $I$ on any input, taking the
same number of internal steps as $\inc$ and $\tes$.  We show that
$M'\sim N''$ by proving that the following relation $𝓢$ is a strong
bisimulation \emph{up to unfolding, store, weakening, strengthening,
  transitivity and context} (a technique unsound in the weak case):
\[
𝓢
\defi \{((∅;M'),(∅;N''))\} ~∪~ \left\{
  \left(
  \left(ℓ↦2 n; \inc,\tes \right),
  \left(∅;\incz, \tesz \right)
  \right)
  ~\forall~ n ∈ \mathbb N
\right\}
\]
This relation uses only one location; it is the union of the singleton
relation $\{((∅;M'),(∅;N''))\}$ and the relation relating
$\left(ℓ↦2 n; \inc,\tes \right)$ to $\left(∅; \incz, \tesz \right)$ for
every integer that can be stored at that location. In the diagram-chasing
arguments for $\S$, essentially a pair of derivatives is proved to be
related under the function
\[ \strongb∘\strongb∘\uptotrans ∘ (\env∪\alloc∪\term∪\weak)^ω \] where
$\uptotrans : {𝓡} ↦ {\starred{𝓡}}$ is the reflexive-transitive closure
function.

  This up-to technique, unsound in the weak case (transitivity is unsound), is powerful enough
  to make the bisimulation considerably  smaller.  Proving that
  the second member of~$𝓢$ progresses to itself (up to $\alloc$) is
  straightforward.  We focus on the following transitions from $M'$ and $N''$:
  \[\def\arraystretch{1.3}
  \begin{array}{c@{}l@{}c@{}l@{}c@{}c@{}c@{}l}   
    (∅,M')  \arrR{1,C,\cod(r)} \,\, &(~r[Γ]\sep ~&Γ,\,&\newloc \,ℓ:=0;\, &C[Γ]&(\inc); & \tes I &~) \eqdef H_1 \\
    (∅,N'') \arrR{1,C,\cod(r)} \,\, &(~r[Δ]\sep ~&Δ,\,& I;I; &C[Δ]&(\incz);& \tesz I &~) \eqdef H_2 
  \end{array}\]
  where $Γ=M',\getset{r}$ and $Δ=N'',\getset{r}$.  We use $\strongb$
  as an up-to technique\footnote{If $\ts{τ}$ is deterministic then $(\ts{τ}𝓡\rts{τ}) ⊆ \strongb(𝓡)$.} twice
so   to run two steps of reduction on both sides:
\[H_1\ts{τ}\ts{τ}H_1'  \hskip .5cm \mbox{  and } \hskip .5cm H_2\ts{τ}\ts{τ}H_2' \; . \]  This way we trigger $νℓ$ and $ℓ:=0$ and obtain two
  configurations $H_1'$ and $H_2'$ that  can be related using a few
  up-to functions:
  \begin{align}
    \label{e:a}
    &(\,r[Γ]⊎(ℓ↦0)\sep  ~Γ,\,C[Γ](\inc); \tes I\,) = H_1'\\
    \label{e:b}
    \weak(\term(\alloc(\env(𝓢))))\quad
    &(\,r[Γ]\sep ~Γ,\,C[Γ](\incz); \tesz I\,)\\
    \label{e:c}
    \term(\alloc(𝓢))\quad
    &(\,r[Δ]\sep ~Δ,\,C[Δ](\incz); \tesz I\,)  =H_2' \enspace.
  \end{align}
  We detail below how we go from \eqref{e:a} to \eqref{e:b}. We write
  $Γ_ℓ \eqdef \inc,\tes$ and $Γ_0 \eqdef \incz,\tesz$, and use
  $-$ as a shorthand for the relation mentioned in the line above it:
  \[\begin{array}{r@{}lcr@{}l}
    (ℓ↦0&\sep Γ_ℓ) &𝓢& (∅&\sep Γ_0)\\
    (ℓ↦0&\sep Γ_ℓ,M') &\env(-)& (∅&\sep Γ_0,M')\\
    (r[Γ]⊎ℓ↦0&\sep Γ_ℓ,Γ) &\alloc(-)& (r[Γ]&\sep Γ_0,Γ)\\
    (r[Γ]⊎ℓ↦0&\sep Γ_ℓ,Γ,C[Γ](\inc);\tes I)) &\term(-)& (r[Γ]&\sep Γ_0,Γ,C[Γ](\incz);\tesz I)\\
    (r[Γ]⊎ℓ↦0&\sep \phantom{Γ_ℓ,\,}Γ,C[Γ](\inc);\tes I)) &\weak(-)& (r[Γ]&\sep \phantom{Γ_0,\,}Γ,C[Γ](\incz);\tesz I)\\
  \end{array}\]
  Going from \eqref{e:b} to \eqref{e:c} is easier:
  \[\begin{array}{lr@{}ll}
    (∅\sep M') ~~&&𝓢&~(∅\sep N'')\\
    (r[Γ]\sep Γ) ~~&\alloc(&-)&~(r[Δ]\sep Δ)\\
    (r[Γ]\sep Γ, C[Γ](\incz);\tesz I) ~~&\term(&-)&~(r[Δ]\sep Δ, C[Δ](\incz);\tesz I)
  \end{array}\]
  We have thus proved that $H_1 \mathrel{f(𝓢)} H_2$ where $f= \strongb ∘
  \strongb ∘ \uptotrans ∘ (\env ∪ \alloc ∪ \term ∪ \weak)^ω$
  is below $t$, and hence $𝓢 \relprogress{} f(𝓢)∪\alloc(𝓢)$.
  To conclude, $𝓢$, as a strong bisimulation up to
  (unfolding, store, weakening, strengthening, transitivity and
  context), is included in $\sim$.

\medskip

The difference between the relation $\R$ in the first proof and the
proofs in \cite{KoutavasW06,envbisim} is that $\R$ only requires
locations that appear free in the tested terms; in contrast, the
relations in \cite{KoutavasW06,envbisim} need to be closed under all
possible extensions of the store, including extensions in which
related locations are mapped onto arbitrary context-closures of
related values. We avoid this thanks to the up-to-store function, by
discarding these extensions immediately after their introduction. The
reason why, both in \cite{KoutavasW06,envbisim} and in the first proof
above, several locations have to be considered is that, with
bisimulations akin to environmental bisimulation, the input for a
function is built using the values that occur in the candidate
relation. In our example, this means that the input for a function can
be a context-closure of $M$ and $N$; hence uses of the input may cause
several evaluations of $M$ and $N$, each of which generates a new
location. In this respect, it is surprising that our second proof
avoids multiple allocations (the candidate relation $𝓢$ only mentions
one location). This is due to the massive combination of up-to
techniques whereby, whenever a new location is created, a double
application of up-to-context (the `double' is obtained from
up-to-transitivity) together with some administrative work (given by
the other techniques) allows us to absorb the location.


\section{Conclusions}
\label{s:conc}
             
In this paper we have studied how to transport the rich theory of `up-to' techniques that
exists for plain (first-order) LTSs and bisimilarity, and rooted in fixed-point theory,
onto languages whose LTS and bisimilarity go beyond the first-order format.  For this
we have considered the $\pi$-calculus, the pure call-by-name $\lambda$-calculus, and a
call-by-value $\lambda$-calculus extended with imperative features.

The approach that we have proposed exhibits fully abstract translations of
the LTSs and bisimilarities of these languages onto first-order LTSs. In this way,
one can directly reuse the large corpus of up-to techniques that are available on
first-order LTSs.  The only exception to this regards basic up-to techniques that are
specific to the new languages, such as up-to-context.  Most important, the approach allows
one to take arbitrarily complex combinations of up-to techniques, whose soundness is
guaranteed by those of the corresponding first-order techniques. Direct proofs of such
combinations, on the source languages, can be long and delicate.  We have given examples
of uses of such combinations.  In particular, the second proof of the example dealt with in
Section~\ref{ss:exa} is, in our opinion, a striking example of the benefits of the up-to
techniques. It is hard to imagine how to the example could be handled without up-to
techniques; compared to similar proofs in the literature and discussed in that section~--- all
of which make use of some forms of up-to techniques~--- the relation employed and the
proof work needed have been significantly reduced due to the large set of up-to techniques
referred to.

The work in this paper can be a further motivation for the development of a comprehensive
formalised library of up-to techniques for first-order LTSs.  Using the approach proposed
in the paper the library could then be applied to a wide range of languages.
Other directions for future work include 
testing the  approach on other languages, for instance languages for mobility with explicit
notions of location (see \cite{Cas01} for a survey), or testing it on other forms of bisimulation 
(e.g., open bisimulation).   
By the time the revision of this paper has been completed, 
one such application of the approach has been made, 
for a calculus with delimited-control operators with dynamic prompt
       generation~\cite{Lenglet:environmentalDelimitedControl}.

We would also like to see if the approach proposed in this paper, based on translations to
first-order models, could be adapted to handle 
the theory of 
unique solutions of equations and contractions \cite{Sangiorgi17,DurierHS17}, which allows
one to implicitly use  up-to techniques, with the goal of avoiding the development of  theories of
equations or contractions that are specific to a particular language or bisimulation.


\section*{Acknowledgements}
We are delighted to be able to contribute to the Festschrift in  honour of Jos Baeten.
We
would like to take this opportunity for heartily thanking him  for having been such an
inspiring figure, both for all his many technical and scientific  contributions and for
his work in favour of the concurrency theory community.

We would like also to thanks the anonymous referees for many useful comments. 
Pous was supported by the European Research
Council (ERC) under the European Union’s Horizon 2020 programme
(CoVeCe, grant agreement No 678157).
Sangiorgi  acknowledges support from the
 MIUR-PRIN project `Analysis of
Program Analyses' (ASPRA, ID: \linebreak[4]\texttt{201784YSZ5\_004}), and from the 
European Research
Council (ERC) Grant DLV-818616 DIAPASoN.

\bibliographystyle{alpha}
\bibliography{references}


\end{document}